\documentclass{article}

\usepackage{graphicx}

\usepackage{amsfonts}
\usepackage{amssymb}
\usepackage{amstext}
\usepackage{amsmath}
\usepackage{enumerate}
\usepackage{algorithm}
\usepackage{algorithmic}

\usepackage{graphicx}

\usepackage{bbm}

\usepackage{amsthm}

\usepackage{thmtools}

\usepackage{verbatim} 

\usepackage[dvipsnames,svgnames,table]{xcolor}
\definecolor{dartmouthgreen}{rgb}{0.05, 0.5, 0.06}
\definecolor{ceruleanblue}{rgb}{0.16, 0.32, 0.75}

\usepackage[colorlinks=true,pdfpagemode=UseNone,citecolor=dartmouthgreen,linkcolor=ceruleanblue,urlcolor=BrickRed,pdfstartview=FitH]{hyperref}
\usepackage{cleveref}

\usepackage{framed}
\usepackage{enumitem}


\newtheorem{theorem}{Theorem}[section]
\newtheorem{fact}[theorem]{Fact}
\newtheorem{lemma}[theorem]{Lemma}
\newtheorem{definition}[theorem]{Definition}

\newtheorem{claim}[theorem]{Claim}

\newtheorem*{problem*}{Problem}
\newtheorem{remark}[theorem]{Remark}
\newtheorem*{remark*}{Remark}

\numberwithin{equation}{section}
\numberwithin{table}{section}

\renewcommand{\preceq}{\preccurlyeq}
\renewcommand{\succeq}{\succcurlyeq}

\renewcommand{\tilde}{\widetilde}

\newcommand{\R}{\ensuremath{\mathbb R}}

\newcommand{\K}{\mathcal{K}}

\newcommand{\poly}{\operatorname{poly}}

\newcommand{\junk}[1]{}

\renewcommand{\L}{{\mathcal L}}

\newcommand{\norm}[1]{\left\lVert#1\right\rVert}

\newcommand{\vertiii}[1]{{\left\vert\kern-0.25ex\left\vert\kern-0.25ex\left\vert #1 \right\vert\kern-0.25ex\right\vert\kern-0.25ex\right\vert}}

\newcommand{\one}{\ensuremath{\mathbbm{1}}}

\def\b1{{\bf 1}}
\def\eps{{\epsilon}}

\def\R{\mathbb{R}}

\def\diag{\operatorname{diag}} 
\def\polylog{\operatorname{polylog}} 

\def\tr{\operatorname{tr}}

\global\long\def\R{\mathbb{R}}
\global\long\def\C{\mathbb{C}}

\newcommand{\inner}[2]{\langle #1, #2 \rangle} 

\DeclareMathOperator{\supp}{supp}

\newcommand{\bigip}[2]{\bigl\langle #1, #2 \bigr\rangle}

\title{Spectral Sparsification by Deterministic Discrepancy Walk}

\author{
  Lap Chi Lau\footnote{University of Waterloo, Email: \href{mailto:lapchi@uwaterloo.ca}{lapchi@uwaterloo.ca}}
  \and
  Robert Wang\footnote{University of Waterloo, Email: \href{mailto:robert.wang2@uwaterloo.ca}{robert.wang2@uwaterloo.ca}}
  \and
  Hong Zhou\footnote{Fuzhou University, Email: \href{mailto:hong.zhou@fzu.edu.cn}{hong.zhou@fzu.edu.cn}}
}

\date{}

\begin{document}

\maketitle

\begin{abstract}
Spectral sparsification and discrepancy minimization are two well-studied areas that are closely related. 
Building on recent connections between these two areas, 
we generalize the ``deterministic discrepancy walk'' framework by Pesenti and Vladu [SODA~23] for vector discrepancy to matrix discrepancy, and use it to give a simpler proof of the matrix partial coloring theorem of Reis and Rothvoss [SODA~20].
Moreover, we show that this matrix discrepancy framework provides a unified approach for various spectral sparsification problems, from stronger notions including unit-circle approximation and singular-value approximation to weaker notions including graphical spectral sketching and effective resistance sparsification.
In all of these applications, our framework produces improved results with a simpler and deterministic analysis.
\end{abstract}

\newpage

\section{Introduction}

The aim of this work is to simplify and unify some previous work on spectral sparsification and discrepancy minimization.
For a better perspective of our results, we will first briefly review the main results and techniques in these two areas, and then see recent connections between them and also some recent motivations.

\subsubsection*{Spectral Sparsification}

Karger~\cite{Kar99} introduced the notion of a cut sparisifer of a graph, which is a sparse graph over its vertex set that approximately preserves the weight of all cuts.
Bencz\'ur and Karger~\cite{BK15} designed a non-uniform random sampling algorithm and proved that any graph on $n$ vertices has a cut sparsifier with $O(n \log n)$ edges.
This is a highly influential result that has many applications in designing fast graph algorithms.

Spielman and Teng~\cite{ST11} introduced a stronger notion called spectral sparsification.
Given a graph $G$, a sparse graph $H$ on the same vertex set is called an $\eps$-spectral sparsifer of $G$ if 
$(1-\eps) L_G \preceq L_H \preceq (1+\eps) L_G$
where $L_G$ and $L_H$ are the Laplacian matrices of $G$ and $H$ respectively\footnote{
Note that a spectral sparsifier is a cut sparsifier but not necessarily vice versa.}.
They proved that any graph on $n$ vertices has an $\eps$-spectral sparsifier with $O( (n \polylog n) / \eps^2)$ edges, and used it to design an important algorithm for solving Laplacian equations~\cite{ST14}.

The stronger notion of spectral sparsification admits the following linear algebraic formulation of the problem.
Given vectors $v_1, \ldots, v_m \in \R^n$ such that $\sum_{i=1}^m v_i v_i^\top = I_n$, find a sparse reweighing $s \in \R_{\geq 0}^m$ such that $\sum_{i=1}^m s(i) \cdot v_i v_i^\top \approx I_n$.
Using this formulation, Spielman and Srivastava~\cite{SS11} proved that any graph on $n$ vertices has an $\eps$-spectral sparsifier with $O( (n \log n)/ \eps^2)$ edges, generalizing the result of Bencz\'ur and Karger, by analyzing a natural non-uniform sampling algorithm using matrix concentration inequalities.
Subsequently, Batson, Spielman and Srivastava~\cite{BSS12} proved the striking result that any graph on $n$ vertices has an $\eps$-spectral sparsifier with only $O( n / \eps^2)$ edges\footnote{
Until now, there is no alternative approach to prove the existence of a linear-sized cut sparsifier, without using the spectral sparsification formulation.}.
Their algorithm is quite different from random sampling, based on an ingenious potential function and a deterministic incremental approach. 

Furthermore, Marcus, Spielman and Srivastava~\cite{MSS15} extended the potential function argument significantly to prove Weaver's conjecture~\cite{Weaver} about a discrepancy problem\footnote{Weaver proved that the discrepancy problem is equivalent to a major open problem in mathematics called the Kadison-Singer problem.  We will describe its interpretation as a discrepancy minimization problem.}: Given vectors $v_1, \ldots, v_m$ such that $\sum_{i=1}^m v_i v_i^\top = I_n$ and $\norm{v_i}^2_2 \leq \eps$ for $1 \leq i \leq m$, there is a partitioning $S_1 \cup S_2 = [m]$ such that $(\frac12 - \sqrt{\eps})^2 \cdot I_n \preceq \sum_{i \in S_j} v_i v_i^\top \preceq (\frac12 + \sqrt{\eps})^2 \cdot I_n$ for $j \in \{1,2\}$. 
The results and techniques for spectral sparsification and Weaver's conjecture have far-reaching consequences with many subsequent works.

\subsubsection*{Discrepancy Minimization}

Discrepancy theory is a broad area with various settings~\cite{Chazelle, Matousek}.
We consider the combinatorial setting where there are $m$ subsets $S_1, \ldots, S_m$ of the ground set $[n]$ with $m \geq n$, and the goal is to find a coloring $s \in \{\pm 1\}^m$ to minimize the maximum discrepancy defined as $\max_{i \in [m]} \big|\sum_{j \in S_i} s(j)\big|$.
A famous result by Spencer~\cite{Spe85} proved that there exists a coloring with discrepancy $O(\sqrt{n \log (m/n)})$, beating the bound obtained by a random coloring.
His proof is non-constructive, based on the pigeonhole principle, and it was a long-standing open question to design an efficient algorithm to find such a coloring.

Bansal~\cite{Ban10} made the breakthrough in obtaining the first polynomial time algorithm to find a coloring matching Spencer's result when $m=n$.
Lovett and Meka~\cite{LM15} simplified Bansal's algorithm and fully recovered Spencer's result.  
Both algorithms in~\cite{Ban10,LM15} are based on some interesting random walks similar to Brownian motion with elegant analyses. 
Rothvoss~\cite{Rot17} provided a beautiful alternative algorithmic approach to recover Spencer's result using convex geometry, reducing the problem to lower bounding the volume of a convex body\footnote{If the volume of a convex body in $\R^n$ is ``large enough'', then one can obtain a point in the convex body with $\Omega(n)$ coordinates being $\{\pm 1\}$ by simply projecting a random Gaussian vector to the convex body.} (a polytope in Spencer's setting).
These results laid the foundation for many subsequent developments in algorithmic discrepancy theory in recent years.

The combinatorial discrepancy problem can also be formulated as a vector discrepancy problem:
Given vectors $v_1, \ldots, v_n \in \R^m$, find a coloring $s \in \{\pm 1\}^n$ to minimize $\norm{\sum_i s(i) \cdot v_i}_{\infty}$. 
A natural generalization is the following matrix discrepancy problem: 
Given matrices $A_1, \ldots, A_n \in \R^{m \times m}$, find a coloring $s \in \{\pm 1\}^n$ to minimize $\norm{\sum_i s(i) \cdot A_i}_{{\rm op}}$.
The matrix Spencer conjecture\footnote{It is not difficult to see that it already generalizes Spencer's result when each matrix is a diagonal matrix.}~\cite{Zou12,Mek14} has generated much interest: 
Given symmetric matrices $A_1, \ldots, A_n \in \R^{n \times n}$,
if $\norm{A_i}_{{\rm op}} \leq 1$ for $1 \leq i \leq n$, then there is a coloring $s \in \{\pm 1\}^n$ with $\norm{\sum_i s(i) \cdot A_i}_{\infty} = O(\sqrt{n})$.
Recently, Bansal, Jiang, and Meka~\cite{BJM23} made significant progress in resolving this conjecture and revealed interesting connections to random matrix theory and free probability~\cite{BBvH23}.

\subsubsection*{Connections}

While the setting of matrix discrepancy minimization seems naturally similar to that of spectral sparsification\footnote{Note that Weaver's conjecture can be restated as finding a coloring $\{\pm 1\}^m$ such that $\norm{\sum_{i=1}^m s(i) \cdot v_i v_i^\top}_{{\rm op}} \leq O(\eps)$ }, 
there are no known concrete connections between the two areas until relatively recently.

In one direction, Reis and Rothvoss~\cite{RR20} were the first to use techniques in algorithmic discrepancy theory to construct linear-sized spectral sparsifiers.
The main technical result in~\cite{RR20} is the following matrix partial coloring theorem\footnote{One advantage of this theorem is that it works for any symmetric matrices, not just for rank one matrices of the form $vv^\top$, and this will be useful for our applications.}:
Given symmetric matrices $A_1, \ldots, A_m \in \R^{n \times n}$ that satisfies $\sum_{i=1}^m |A_i| \preceq I_n$, there is a partial fractional coloring $x \in [-1,+1]^m$ such that there are $\Omega(m)$ coordinates in $\{\pm 1\}$ and $\| \sum_{i=1}^m x(i) \cdot A_i \|_{{\rm op}} \leq O(\sqrt{n/m})$.
This matrix partial coloring theorem can be applied recursively to construct a sparse reweighting for spectral sparsification~\cite{RR20}, which we will explain in \autoref{s:SV-sparsification}.
To prove the matrix partial coloring theorem, Reis and Rothvoss used the convex geometric approach and bounded the volume of the operator norm ball $\{x \in \R^m : \| \sum_{i=1}^m x(i) \cdot A_i \|_{{\rm op}} \leq O(\sqrt{n/m}) \}$.
To lower bound the volume of the norm ball, they used the potential function in~\cite{BSS12} as a barrier function to show that a guided Gaussian random walk will stay in the norm ball with high probability.
Recently, this approach was extended by Jambulapati, Reis and Tian~\cite{JRT24} to construct linear-sized spectral sparsifier that are degree preserving.
Furthermore, Sachdeva, Thudi, and Zhao~\cite{STZ24} used the result in~\cite{BJM23} for the matrix Spencer problem to construct better spectral sparsifiers for Eulerian directed graphs.

In the other direction, Bansal, Laddha and Vempala~\cite{BLV22} and Pesenti and Vladu~\cite{PV23} used the potential functions in spectral sparsification for discrepancy minimization.
Both of these algorithms can be understood as ``deterministic walks'', tracking the discrepancy of a continuously evolving partial fractional coloring guided by a potential function.
The potential function used in~\cite{BLV22} is similar to that in~\cite{BSS12}, while the potential function used in~\cite{PV23} is based on a regularized optimization formulation developed by Allen-Zhu, Liao and Orecchia~\cite{AZLO15} for spectral sparsification.
They both showed that their approach provides a unifying algorithmic framework to recover many best known results in discrepancy minimization.

It is emerging from these recent connections that the potential functions in~\cite{BSS12,AZLO15} and the Brownian/deterministic discrepancy walks in~\cite{Ban10,LM15,RR20,BLV22,PV23} are the two underlying components for both spectral sparsification and discrepancy minimization.

\subsubsection*{Motivations}

Besides being a natural goal in developing a simple and unifying algorithmic framework for both spectral sparsification and discrepancy minimization,
there are recent applications that require better techniques for constructing stronger notions of spectral sparsifiers.
Motivated by problems related to solving directed Laplacian equations~\cite{CKPRSV17} and log-space derandomization of directed connectivity~\cite{AKMPSV20,APPSV23}, several new notions of spectral sparsification for directed graphs were proposed, including standard approximation~\cite{CKPRSV17}, unit-circle approximation~\cite{AKMPSV20} and singular-value approximation~\cite{APPSV23}.
We will present the formal definitions and mention some properties of these notions in \autoref{section: sparsifiers}.

One common requirement of these notions is that the sparsifiers need to be degree preserving.
It is not known how to adapt the standard techniques such as independent random sampling in~\cite{SS11} and the potential function based incremental algorithm in~\cite{BSS12,AZLO15} to satisfy this requirement.
The constructions in~\cite{CKPRSV17,APPSV23} and a recent improvement by Sachdeva, Thudi and Zhao are based on the short cycle decomposition technique developed in~\cite{CGP+23}, which is a dependent rounding method that samples alternating edges in even cycles to ensure degree preservation.

There are motivations in replacing the cycle decomposition technique by discrepancy walks.
An open problem is to determine whether there are linear-sized sparsifiers for these stronger notions of spectral sparsification~\cite{APPSV23}.
The cycle decomposition technique inherits some polylogarithmic losses and it is not even known how to use it to construct linear-sized spectral sparsifiers in the standard setting~\cite{BSS12},
while the discrepancy approach was used successfully to construct linear-sized degree-preserving spectral sparsifiers~\cite{JRT24}.
Another motivation is to develop a unifying approach for spectral sparsification and discrepancy minimization, to consolidate our understanding for more challenging problems in these fields.

\subsection{Our Results}

We generalize the deterministic discrepancy walk approach by Pesenti and Vladu~\cite{PV23} for the vector discrepancy setting to the matrix discrepancy setting, and obtain a considerably simpler proof of the matrix partial coloring theorem of Reis and Rothvoss~\cite{RR20}.
Moreover, we show that this provides a unifying framework to derive simpler and sparser constructions for various spectral sparsification problems.
One additional advantage is that this gives the first deterministic algorithms for all these applications, which also leads to more elementary proofs and simpler calculations bypassing the matrix concentration inequalities and probabilistic analyses in previous work.

\subsubsection*{Matrix Partial Coloring and Matrix Sparsification}

The main technical result is a deterministic version of the matrix partial coloring theorem by Reis and Rothvoss~\cite{RR20}, augmented with linear constraints as was done in~\cite{JRT24}. 
The linear constraints provide extra flexibility, such as preserving degrees exactly, that are crucial in all applications in this paper.

\begin{theorem}[Deterministic Matrix Partial Coloring]\label{lemma:deterministic-partial-coloring}
    Let $A_1, \dots, A_m \in \R^{n \times n}$ be symmetric matrices such that $\sum_{i=1}^m |A_i| \preccurlyeq I_d$. 
    Let $\mathcal{C} \subseteq \R^m$ be a set of good partial fractional colorings defined as
    \[
    \mathcal{C} := \left\{x \in \R^m : \norm{\sum_{i=1}^m x(i) \cdot A_i}_{\rm op} \leq 16 \sqrt{\frac{n}{m}} \right\}.
    \]
    In addition, let $\mathcal{H} \subseteq \R^m$ be a linear subspace of dimension $c \cdot m$ for some constant $c \geq \frac45$.
    There is a deterministic polynomial time algorithm that returns a partial fractional coloring $x \in [-1, 1]^m$ such that
    \[
        x \in \mathcal{C} \cap \mathcal{H} \quad \text{and} \quad |\{i \in [m] : x_i = \pm 1\}| = \Omega(m).
    \]
\end{theorem}

We simplify the proof of Reis and Rothvoss' result both conceptually and technically.
In~\cite{RR20}, a random walk guided by the potential function in~\cite{BSS12} was used to lower bound the Gaussian volume of the operator norm ball $\mathcal{C}$\footnote{This is actually an over-simplication of what was done in~\cite{RR20}, which is only made precise in~\cite{JRT24} through nontrivial convex geometric arguments. The original statements in~\cite{RR20} are more involved.}, which implies the existence of a good partial coloring through the convex geometric argument in~\cite{Rot17}.
In contrast, we directly use a deterministic discrepancy walk to construct a good partial coloring, bypassing all convex geometric concepts.
Also, we used the regularized optimization framework developed in~\cite{AZLO15,PV23} instead of the closely related potential function in~\cite{BSS12}.
This allows us to bound the matrix discrepancy through a standard step-by-step optimization process (like Newton steps), without manually shifting the barrier and dealing with some complicated higher order terms as in~\cite{RR20}.
This framework is more systematic and leads to conceptually simpler analysis with considerably easier calculations.
We will elaborate more about this point in \autoref{s:overview}.

By the similar reduction as in~\cite{RR20}, we obtain a deterministic matrix sparsification result with additional subspace requirement.

\begin{theorem}[Deterministic Matrix Sparsification] \label{theorem:sparsification-alg}
    Given positive semidefinite matrices $A_1, A_2,...A_m\in \mathbb{R}^{n\times n}$ such that $\sum_{i=1}^m A_i\preceq I_n$ and a linear subspace $\mathcal{H} \subseteq \mathbb{R}^m$ of dimension $m-O(n)$, there is a polynomial time deterministic algorithm to construct a sparse reweighting $s \in \R_{\geq 0}^m$ with $|\supp(s)| = O(n / \eps^2)$ such that $s - \vec{\one}_m\in \mathcal{H}$ and 
\[        
\bigg\| \sum_{i=1}^m s(i) \cdot A_i - \sum_{i=1}^m A_i \bigg\|_{\rm op} \leq \eps.
\]
\end{theorem}

\subsubsection*{Stronger Notions of Spectral Sparsification}

We demonstrate how to use~\autoref{theorem:sparsification-alg} to obtain improved constructions for stronger spectral sparsifiers.
One easy corollary of \autoref{theorem:sparsification-alg} is that we can deterministically find degree-preserving $\eps$-spectral sparsifiers of size $O(n/\eps^2)$, which derandomizes a result in~\cite{JRT24} (see \autoref{s:spectral-subspace}).

The unit-circle (UC) approximation introduced in~\cite{AKMPSV20} satisfies the property that it preserves the spectrum of all the powers of the random walk matrix of a directed graph (see \autoref{section: sparsifiers} for the formal statement),
and it was used to study low-space Laplacian solvers for undirected and Eulerian directed graphs.
For undirected graphs, this definition simplifies to the following (see \autoref{lemma: undirected-uc-sparsifier}):
$H$ is an $\eps$-UC approximation of $G$ if (i) $D_H=D_G$, (ii) $D_H-A_H \approx_\eps D_G-A_G$, and (iii) $D_H+A_H \approx_\eps D_G+A_G$, where $D$ denotes the diagonal degree matrix, $A$ denotes the adjacency matrix, and $\approx_\eps$ denotes the usual $\eps$-spectral approximation of matrices.
We show that a simple reduction to \autoref{theorem:sparsification-alg} gives linear-sized UC-sparsifiers in deterministic polynomial time, where such linear-sized sparsifiers were not known to exist previously\footnote{
We remark that no explicit constructions of UC-sparsifiers for general directed graphs were given in~\cite{AKMPSV20}, not even for general undirected graphs, as it is enough to only sparsify squares of graphs for their applications.
In an earlier work~\cite{CCLPT15} that studies similar properties for undirected graphs, there were also no explicit constructions for general undirected graphs, but just for squares of graphs.}  
and can be seen as a strengthening of the linear-sized spectral sparsifiers in~\cite{BSS12}.

\begin{theorem}[Linear-Sized UC-Sparsifiers of Undirected Graphs] \label{theorem: undirected-uc-sparsifier}
There is a polynomial time deterministic algorithm to compute an $\eps$-UC-approximation with $O(n/\eps^2)$ edges for any undirected graph on $n$ vertices.
\end{theorem}

The singular-value (SV) approximation introduced in~\cite{APPSV23} is a strictly stronger notion than UC-approximation, and satisfies the stronger property that it is preserved under arbitrary products (see \autoref{section: sparsifiers} for the formal statement), which has applications in approximating higher powers of directed random walks and solving directed Laplacian systems.
Using the cycle decomposition technique in~\cite{CGP+23} and expander decomposition, it was proved in~\cite{APPSV23} that any directed graph has an $\eps$-SV-sparsifier (and hence an $\eps$-UC-sparsifier) with $O((n \log^{12}{n})/\eps^2)$ edges.
We show that a simple reduction to \autoref{theorem:sparsification-alg} gives linear-sized $\eps$-SV-sparsifier when the bipartite lift is an expander graph.
This simplifies and improves the corresponding $O(n \log^2 n)$ bound obtained using the cycle decomposition technique in~\cite{APPSV23}, and implies the following improved bound for general graphs through the same standard expander decomposition technique as in~\cite{APPSV23}.

\begin{theorem}[Improved SV-Sparsifiers of Directed Graphs] \label{theorem: SV-sparsfication-general}
There is a polynomial time deterministic algorithm to compute an $\eps$-SV-sparsifier with $O((n\log^5{n})/\eps^2)$ edges for any directed graph on $n$ vertices.
\end{theorem}

\subsubsection*{Graphical Spectral Sketches and Effective Resistance Sparsifiers}

We also find applications of \autoref{theorem:sparsification-alg} to obtain improved constructions for weaker notions of spectral sparsification.
An $\eps$-spectral sketch of a graph $G$ is a randomized data structure that, for any vector $x$, preserves the quadratic form $x^\top L_G x$ within a $1\pm \eps$ multiplicative factor with high probability~\cite{ACK+16,JS18}. 
Chu, Gao, Peng, Sachdeva, Sawlani, and Wang~\cite{CGP+23} considered the notion of a graphical spectral sketch, a spectral sketch that is a reweighted subgraph of $G$. 
They proved the interesting result that any graph has a graphical $\eps$-spectral sketch with only $O( (n \polylog n) / \eps)$ edges, with a better dependency on $\eps$ than an $\eps$-spectral sparsifier which requires $\Omega(n / \eps^2)$ edges~\cite{BSS12}.

In this work, we consider a deterministic version of graphical $\eps$-spectral sketch, which is a sparsifier that $\eps$-approximates an input graph's Laplacian quadratic forms with respect to a specified set of constraint vectors $\mathcal K$ (see \autoref{s:sketch} for a formal definition).
Using the discrepancy approach, we obtain a deterministic polynomial time algorithm to construct sparser graphical $\eps$-spectral sketches.

\begin{theorem}[Improved Graphical Spectral Sketching] \label{theorem: deterministic-eps-sketch}
There is a polynomial time deterministic algorithm to construct a graphical $\eps$-spectral sketch with $O\big(\big(n \log^{3}{n} \cdot \max\big\{1,\sqrt{\log(|\K|/n)}\big\}\big)/\eps\big)$ edges for any unweighted undirected graph on $n$ vertices, any set of vectors $\K$ in $\R^n$, and any $\eps > 0$. 
\end{theorem}

Graphical spectral sketches were also used to construct effective resistance sparsifiers in~\cite{CGP+23}, which preserve all pairwise effective resistances between vertices in a graph.
Note that this is also a notion weaker than spectral sparsification, as the effective resistance between two vertices $i$ and $j$ in $G$ can be written as $b_{i,j}^\top L_G^\dagger b_{i,j}$ where $b_{i,j} := \vec{\one}_i - \vec{\one}_j$ and $L_G^\dagger$ is the pseudo-inverse of the Laplacian matrix $L_G$. 
Following the reduction in~\cite{CGP+23}, we can also extend \autoref{theorem: deterministic-eps-sketch} to produce sparser effective resistance sparsifiers.

\begin{theorem}[Improved Effective Resistance Sparsification] \label{cor:resistance-sparsifier}
There is a polynomial time deterministic algorithm to construct a reweighted subgraph $\Tilde{G}$ with $O((n\log^{3.5}{n})/\eps)$ edges such that
    \[
        (1-\eps) \cdot b_{i,j}^\top L_G^\dagger b_{i,j} 
	\leq b_{i,j}^\top L_{\Tilde{G}}^\dagger b_{i,j} 
	\leq (1+\eps) \cdot b_{i,j}^\top L_G^\dagger b_{i,j}
    \]
for all $i,j \in V(G)$ for any unweighted undirected graph $G$.
\end{theorem}

Both \autoref{theorem: deterministic-eps-sketch} and \autoref{cor:resistance-sparsifier} improve the previous results by many $\log n$ factors: the graphical spectral sketches and the effective resistance sparsifiers in~\cite{CGP+23} have $\Omega( (n \log^{16} n) / \eps)$ edges even assuming optimal short cycle decomposition and optimal expander decomposition.

\subsection{Technical Overview} \label{s:overview}

The proof of \autoref{lemma:deterministic-partial-coloring} is quite straightforward conceptually, once the potential function is fixed. 
Building on the regularized optimization framework for vector discrepancy by Pesenti and Vladu~\cite{PV23} and for spectral sparsification by Allen-Zhu, Liao, and Orecchia~\cite{AZLO15}, we use the potential function
\[
\Phi(x) := \max_{M \in \Delta_n} \bigip{A(x)}{M} + \frac{2}{\eta} \tr\big(M^{\frac12}\big),
\]
where $A(x) := \sum_{i=1}^m x(i) \cdot A_i$ is a shorthand of the current solution, $\Delta_n := \{M \succeq 0 \mid \tr(M)=1\}$ is the set of density matrices , and $\eta$ is a parameter.
Note that the first term is simply $\lambda_{\max}( A(x) )$, the maximum eigenvalue of our current solution.
The term $2\tr(M^{\frac12})$ is known as the $\ell_{1/2}$-regularizer, and the potential function $\Phi(x)$ can be seen as a regularized smooth proxy of $\lambda_{\max}(A(x))$.
We remark that the optimizer $M$ can be computed and the closed form characterization of $\Phi(x)$ is very similar to the potential function used in~\cite{BSS12}; see \autoref{l:potential} and (\ref{eq:Phi}) in the appendix.

To find the partial fractional coloring, we start from the initial point $x_0 = 0$.
In each iteration $t \geq 1$, the natural step is to find a small perturbation $y \in \R^m$ and set $x_{t} = x_{t-1} + y$, so that $\norm{x_{t}}_2 > \norm{x_{t-1}}_2$ and $\Phi(x_t)$ is not much larger than $\Phi(x_{t-1})$.
We prove in \autoref{l:potential} that
\[
\Phi(x+y) - \Phi(x) \leq \tr(M A(y)) + c \eta \tr\big( M^{\frac12} A(y) M^{\frac12} A(y) M^{\frac12} \big),
\]
for some $|c| \leq 2$ when $\norm{y}$ is small enough where $M$ is the optimizer.
This bound can be derived from~\cite{AZLO15} using concepts from mirror descent and Bregman divergence.
We provide a simpler proof of this bound using elementary convexity argument and Taylor expansion (in which we use a second-order approximation from~\cite{RR20}), bypassing these more abstract concepts from optimization; see \autoref{a:omitted}.

To choose the updated perturbation $y$, we use standard arguments in discrepancy minimization to restrict to the subspace that satisfies:
\begin{enumerate}\setlength{\itemsep}{0pt} 
\item $y$ has support only on ``active'' coordinates, those coordinates in $x_{t-1}$ that are not in $\{\pm1\}$ yet,
\item $y \in \mathcal H$ so that $x_t \in \mathcal H$ assuming $x_{t-1} \in \mathcal H$, 
\item $y \perp x_{t-1}$ so that $\norm{x_t}_2^2 = \norm{x_{t-1}}_2^2 + \norm{y}_2^2$ is increasing, in order to upper bound the number of iterations, 
\item the linear term $\tr(M A(y))$ is zero, and 
\item the second order term $\tr( M^{\frac12} A(y) M^{\frac12} A(y) M^{\frac12} )$ is small.
\end{enumerate}
The nontrivial step is to bound the second order term, for which we use the idea in~\cite{PV23} to rewrite it as the quadratic form $y^T N y$ of a matrix, where $N_{ij} = \tr( M^{\frac12} A_i M^{\frac12} A_j M^{\frac12} )$.
With this formulation, we just need to argue that $N$ has a large eigenspace with small eigenvalues.
To do so, we show that a large principal submatrix of $N$ has small trace (in which we borrow another lemma from~\cite{RR20}), from which we conclude the theorem by Cauchy interlacing.

For the applications in spectral sparsification, the linear subspace $\mathcal H$ in the discrepancy framework provides much flexibility in incorporating additional constraints, such as the degree constraints for degree-preserving sparsifiers.
We believe that this perturbation step in the discrepancy framework is a more versatile and powerful dependent rounding method than the cycle decomposition technique.
We demonstrate this (together with the recent works~\cite{JRT24,STZ24}) by providing simpler and sparser constructions using the discrepancy approach, for all applications previously using the cycle decomposition technique, including degree-preserving sparsifiers, Eulerian sparsifiers, graphical spectral sketches, effective resistance sparsifiers, and singular-value sparsifiers.
The reductions to matrix partial coloring for UC-approximation in \autoref{theorem: undirected-uc-sparsifier} and SV-approximation in \autoref{theorem: SV-sparsfication-general} are very simple and direct, while the reduction for graphical spectral sketches in \autoref{theorem: deterministic-eps-sketch} requires some problem-specific insights.

We emphasize that almost all technical ideas used in this work were known in the literature, but we put them together in a more streamlined manner, considerably simplifling many calculations as well as bypassing many advanced concepts (such as matrix concentration inequalities, probabilistic analyses, convex geometry, mirror descent and Bregman divergence, etc).
Our proofs are basically self-contained without using any blackbox (except two basic lemmas from~\cite{RR20}), and are arguably easier than~\cite{BSS12,AZLO15} even in the basic setting (without keeping track of two potential functions).
Also, it is the same approach that works for both discrepancy minimization (as shown in~\cite{PV23}) and spectral sparsification.
We believe that the simplifications and the unification consolidate our understanding and will stimulate further progress.

\section{Background} \label{s:preliminaries}

Given two functions $f, g$, we use $f \lesssim g$ to denote the existence of a positive constant $c > 0$, such that $f \le c \cdot g$ always holds.
For positive integers $n$, we use $[n]$ to denote the set $\{1, 2, \dots, n\}$.

\subsection{Linear Algebra} \label{s:linear-algebra}

Throughout this paper, we use subscripts to index different objects, $v_1, v_2, \dots$ for vectors and $A_1, A_2, \dots$ for matrices.
To avoid confusion, we write the $i$-th entry of a vector $v$ as $v(i)$, and write the $(i,j)$-th entry of a matrix $A$ as $A(i,j)$. 
Given matrices $A_1, \dots, A_m$ and a vector $y \in \R^m$, we will use the shorthand $A(y) := \sum_{i=1}^m y(i) \cdot A_i$.

We write $\vec{\one}$ as the all-one vector and write $\vec{\one}_S$ as the indicator vector of a set $S$.
For a vector $v \in \R^n$ and a set $S \subseteq [n]$, we write $v|_S$ as the $|S|$-dimensional vector restricting $v$ to the coordinates in $S$.
Given a vector $v \in \R^n$, we write $\diag(v) \in \R^{n \times n}$ as the diagonal matrix with entries from $v$ on the diagonal.

Let $M \in \R^{n \times n}$ be a real symmetric matrix
and $M = \sum_{i=1}^n \lambda_i v_i v_i^\top$ be its eigen-decomposition 
with real eigenvalues and orthonormal basis $\{v_i\}_{i=1}^n$.
The largest eigevalue of $M$ is written as $\lambda_{\max}(M) := \max_{i \in [n]} \lambda_i$
and the operator norm of $M$ is written as $\norm{M}_{\rm op} := \max_{i \in [n]} |\lambda_i|$.
We denote $M^\dagger := \sum_{\lambda_i \neq 0} \lambda_i^{-1} v_i v_i^\top$ as the pseudoinverse of $M$ 
and denote $|M| := \sum_{i=1}^n |\lambda_i| v_i v_i^\top $ as the matrix obtained by taking the absolute value over all eigenvalues of $M$.
We say that $M$ is positive semidefinite, denoted by $M \succeq 0$, if all eigenvalues of $M$ are non-negative.
We denote $\Delta_n := \{M \succcurlyeq 0 \mid \tr(M) = 1 \} \subseteq \R^{n \times n}$ as the set of all $n$-dimensional density matrices,
where $\tr(M)$ is defined as the sum of the diagonal entries of $M$.
It is well-known that $\tr(M) = \sum_{i=1}^n \lambda_i$, 
and so the following inequality follows from a simple application of Cauchy-Schwartz.

\begin{lemma} \label{l:root-density}
$\tr(M^{\frac12}) \leq \sqrt{n}$ for any $M \in \Delta_n$.
\end{lemma}

The following three basic facts will also be used.
One is that the maximum eigenvalue of a matrix $A$ can be formulated as $\lambda_{\max}(A) = \max_{M \in \Delta_n} \inner{A}{M}$.
Another is that if $A \succeq 0$ then
$\tr(AB) \leq \norm{B}_{\rm op} \cdot \tr(A)$.
The last one is the well-known Cauchy interlacing theorem.

\begin{theorem}[Cauchy Interlacing Theorem] \label{t:interlacing}
Let $A \in \R^{n \times n}$ be a symmetric matrix with eigenvalues $\alpha_1 \leq \cdots \leq \alpha_n$ and $B \in \R^{m \times m}$ be a principal submatrix of $A$ with eigenvalues $\beta_1 \leq \cdots \leq \beta_m$.
For any $1 \leq i \leq m$,
\[\alpha_i \leq \beta_i \leq \alpha_{n-m+i}.\] 
\end{theorem}

We will also need the following two lemmas from~\cite{RR20}.
The first one can be established by applying the Cauchy-Schwarz inequality twice.

\begin{lemma}[{\cite[Lemma 10]{RR20}}] \label{l:RR}
    Let $A, B, C$ be symmetric matrices with $A, B \succeq 0$.  Then 
\[\tr(ACBC) \leq \tr(A|C|) \cdot \tr(B|C|).\]
\end{lemma}

The second one is a second-order approximation of the trace inverse function under perturbations.

\begin{lemma}[{\cite[Lemma 11]{RR20}}] \label{l:RR2}
    Let $A, B$ be symmetric matrices with $A \succ 0$ and $\norm{\eta A^{-1} B}_{\rm op} \leq \frac12$ for some $\eta > 0$.
    Then there is a value $c \in [-2,2]$ so that
    \[
        \tr\big( (A - \eta B)^{-1} \big) = \tr(A^{-1}) + \eta \tr(A^{-1} B A^{-1}) + c \eta^2 \tr(A^{-1} B A^{-1} B A^{-1}).
    \]
\end{lemma}

\subsection{Spectral Graph Theory}

Let $G = (V,E,w)$ be an undirected graph with weight $w_e \in \R_+$ on each edge $e \in E$. 
For two vertices $u, v \in V$, we write $u \sim v$ to denote that $u$ and $v$ are neighbors.
Given a vertex $v\in V$, we use $d_G(v) := \sum_{u: u \sim v} w_{uv}$ to denote its weighted degree, and we will refer to its number of neigbors as its combinatorial degree. 

Let $A_G\in \mathbb{R}^{n\times n}$ be the adjacency matrix of $G$, with $A_G(u,v) = w_{uv}$ if $uv\in E$ and $0$ otherwise. 
Let $D_G$ be the diagonal degree matrix of $G$ with $D_G(v,v) = d_G(v)$. 
The Laplacian matrix of $G$ is defined as $L_G := D_G - A_G$.
It is well-known that the Laplacian matrix is positive semidefinite,
and can be written as the sum of rank-one matrices $L_G = \sum_{e\in E} w_e b_eb_e^\top$, where $b_e = \vec{\one}_u - \vec{\one}_v$ is the (signed) incidence vector of an edge.
For a connected graph $G$, the random walk matrix is defined as $W_G := D_G^{-1} A_G$, which is a row-stochastic matrix with row sums equal to one.
The normalized Laplacian matrix is defined as $\mathcal{L}_G := D_G^{-1/2} L_G D_G^{-1/2}$, where all eigenvalues of $\mathcal{L}_G$ are between $0$ and $2$. 
Note that $\mathcal{L}_G$ and $W_G$ are similar matrices with the same eigenvalues.
We may drop the subscript $G$ from our matrices when the context is clear.

One class of graphs that will be of particular importance for our sparsification results is the expander graphs. 
Let $\lambda_i(\mathcal{L})$ be the $i$-th smallest eigenvalue of $\mathcal{L}$.  
For any scalar $\lambda$, we say that $G$ is a $\lambda$-expander if $\lambda_2(\mathcal{L}_G)\geq \lambda$.
We will often show that sparsifiers are easier to construct for expander graphs, and then apply the following well-known expander decomposition result (see e.g.~\cite{KVV04}) to construct sparsifiers for general graphs. 

\begin{theorem}[Expander Decomposition]  \label{fact: expander-decomp}
Given an unweighted undirected graph $G$ with $n$ vertices, 
there is a polynomial time deterministic algorithm to decompose its edges into subgraphs $G_1, G_2,...G_k$ for some $k$, such that each $G_i$ is an $\Omega(1/\log^2{n})$-expander and each vertex is contained in at most $O(\log{n})$ subgraphs.
\end{theorem}

We will use the following simple facts about spectral properties of bipartite graphs.

\begin{fact}\label{lemma: symmetric-spectrum}
For any bipartite graph $G$ on $2n$ vertices, $\lambda_{2n-i+1}(\L_G) = 2-\lambda_i(\L_G)$ for $i \in [n]$.
\end{fact}

\begin{fact}\label{lemma: max-eigenvalue-bipartite}
A graph $G$ is bipartite if and only if $\lambda_{\max}(\mathcal{L}_G)=2$.
\end{fact}

We will also consider sparsfying random walks on directed graphs.
Given a directed graph $G$ with adjacency matrix $A$, let $D_{\rm in}$ and $D_{\rm out}$ be the diagonal in-degree and out-degree matrices of $G$ respectively. 
The random walk matrix of $G$ is defined as $W = D_{\rm out}^{-1}A$. 
We say that $G$ is Eulerian if $D_{\rm in} = D_{\rm out}$.

\subsection{Spectral Sparsification for Directed Graphs} \label{section: sparsifiers}

We follow the presentation in~\cite{APPSV23} to define the new notions of spectral sparsification for directed graphs~\cite{CKPRSV17,AKMPSV20,APPSV23} in a unifying way.
The following is a definition of matrix approximation for general non-symmetric matrices, in which the aim is to preserve $x^\top A y$ for all vectors $x$ and $y$ (but not just quadratic form $x^\top A x$ as in the symmetric case) with respect to some appropriate error matrices.

\begin{definition}[Matrix Approximation] \label{def:matrix-approx} 
Let $A, \Tilde{A} \in \C^{n \times n}$ and $E, F \in \C^{n \times n}$ be positive semidefinite matrices.
We say that $\Tilde{A}$ is an $\eps$-matrix approximation of $A$ with respect to error matrices $E$ and $F$ if any of the following equivalent conditions hold:
\begin{enumerate}
\item  $|x^* (A-\Tilde{A})y| \leq \frac{\eps}{2} (x^* Ex + y^* Fy)$ for all $x,y\in \C^n$, 
\item $\big\|E^{\dagger/2}(A-\Tilde{A})F^{\dagger/2}\big\|_{\rm op} \leq \eps$ and $\ker(E)\subseteq \ker((A-\Tilde{A})^\top)$ and $\ker(F)\subseteq \ker(A-\Tilde{A})$.
\end{enumerate}
Here $\ker(M)$ denotes the (right) kernel of a matrix $M$.
\end{definition}

\subsubsection*{Standard Approximation}

Cohen, Kelner, Peebles, Peng, Rao, Sidford and Vladu~\cite{CKPRSV17} introduced the following definition of spectral sparsification for directed graphs, and used it to design fast algorithms for solving directed Laplacian equations.

\begin{definition}[Standard Approximation of Matrices and of Graphs] \label{def:standard-approx}
Let $A, \Tilde{A} \in \C^{n \times n}$.
We say that $\Tilde{A}$ is an $\eps$-standard approximation of $A$ if $\Tilde{A}$ is an $\eps$-matrix-approximation of $A$ with respect to the error matrices $E= F= D-\frac{1}{2}(A+A^*)$, where $D := \frac{1}{2}(D_{A} + D_{A^*})$ and $D_{A}, D_{A^*}$ are the diagonal matrix of the row sums of $A, A^*$ respectively. 

Let $G$ and $\Tilde{G}$ be directed graphs with adjacency matrices $A$ and $\Tilde{A}$.
We say that $\Tilde{G}$ is an $\eps$-standard approximation of $G$ if $\Tilde{A}$ is an $\eps$-standard approximation of $A$.
Note that $D = \frac12 (D_{\rm out} + D_{\rm in})$ in this case and thus the error matrix is exactly the Laplacian matrix of the underlying undirected graph of $G$. 
\end{definition}

Note that as the kernel of any Laplacian matrix contains the all-ones vector, the kernel condition in \autoref{def:matrix-approx} implies that standard approximation is always degree-preserving, 
as $(A-\Tilde{A})\vec{\one} = 0$ implies the outdegrees are the same
and $(A-\Tilde{A})^\top \vec{\one} = 0$ implies the indegrees are the same.
Thus, even for undirected graphs, it is a stronger definition than the spectral approximation of Spielman and Teng. 
Recently, Sachdeva, Thudi, and Zhao~\cite{STZ24} give a polynomial time algorithm to construct an $\eps$-standard sparsifier with $\hat{O}( (n\log^2{n}) / \eps^2 )$ edges (ignoring $\poly \log \log n$ factors) for Eulerian directed graphs, using the discrepancy result in~\cite{BJM23} for the Matrix Spencer problem.

\subsubsection*{Unit-Circle Approximation}

Ahmadinejad, Kelner, Murtagh, Peebles, Sidford, and Vadhan~\cite{AKMPSV20} introduced a stronger notion called unit-circle (UC) approximation, and used it to design deterministic low-space algorithm for estimating random walk probabilities of Eulerian directed graphs.
 
\begin{definition}[Unit-Circle Approximation] \label{def: uc-approx}
Let $G$ and $\Tilde{G}$ be directed graphs with adjacency matrices $A$ and $\Tilde{A}$ respectively.
We say that $\Tilde{G}$ is an $\eps$-UC approximation of $G$ if for all $z \in \C$ with $|z|=1$, the matrix $z\Tilde{A}$ is an $\eps$-standard approximation of the matrix $zA$. 
\end{definition}

It is immediate from the definition that UC approximation implies standard approximation by taking $z=1$.
Suppose $\Tilde{G}$ is an $\eps$-UC approximation of $G$, and $W$ and $\Tilde{W}$ are the random walk matrices of $G$ and $\Tilde{G}$ respectively.
An interesting property of UC approximation is that $\Tilde{W}^k$ is an $O(\eps)$-UC approximation of $W^k$ for all integers $k \geq 1$. 
In addition to preserving powers, UC approximation preserves a stronger property called the cycle lift of random walk matrices, which was a key property in analyzing the low-space Laplacian solver; see Lemma~1.4 in~\cite{AKMPSV20} for these two properties.

For undirected graphs, the adjacency matrix has only real eigenvalues, and it turns out, this means it suffices to check the condition in \autoref{def: uc-approx} for only $z \in \{\pm 1\}$. 
Thus UC approximation admits a simple definition yet still carries strong properties that are not captured by standard spectral approximations.

\begin{lemma}[UC Approximation for Undirected Graphs {\cite[Lemma 3.7]{AKMPSV20}}] \label{fact: uc-properties-undirected}
Let $G$ and $\Tilde{G}$ be undirected graphs with adjacency matrices $A$ and $\Tilde{A}$ respectively.
Then $\Tilde{A}$ is an $\eps$-UC-approximation of $A$ if and only if $\Tilde{A}$ is an $\eps$-standard approximation of $A$ and $-\Tilde{A}$ is an $\eps$-standard-approximations of $-A$.
\end{lemma}

For directed graphs, however, UC approximation is not easy to certify, and it was not known how to construct UC sparsifiers directly for arbitrary directed graphs. 
In~\cite{AKMPSV20}, UC sparsifiers are only constructed for squares of Eulerian graphs, which is sufficient for the purpose of designing Eulerian Laplacian solvers using small space.

\subsubsection*{Singular Value Approximation}

More recently, Ahmadinejad, Peebles, Pyne, Sidford, and Vadhan~\cite{APPSV23} proposed an even stronger notion called singular-value (SV) approximation, and used it to design nearly linear time algorithms for estimating stationary probabilities of general directed graphs.

\begin{definition}[Singular Value Approximation] \label{def:SV-approx}
Let $G$ and $\Tilde{G}$ be directed graphs with adjacency matrices $A$ and $\Tilde{A}$ respectively.
We say that $\Tilde{G}$ is an $\eps$-SV approximation of $G$ if $\Tilde{A}$ is an $\eps$-matrix approximation of $A$ with respect to the error matrices $E = D_{\rm out} - AD_{\rm in}^{\dagger}A^\top$ and $F = D_{\rm in} - A^\top D_{\rm out}^{\dagger}A$, where $D_{\rm out}$ and $D_{\rm in}$ are the diagonal outdegree and indegree matrices of $G$. 
We may also say that $\Tilde{A}$ is an $\eps$-SV approximation of $A$ when it is clear from the context.
\end{definition}

It was proved in~\cite{APPSV23} that SV approximation implies UC approximation and is strictly stronger.
One useful property of SV approximation is being preserved under arbitrary matrix product operations: Given two sequences of directed random walk matrices $W_1,...W_k$ and $\Tilde{W}_1,...\Tilde{W}_k$, if each $\Tilde{W}_i$ is an $\eps$-SV approximation of $W_i$, then the matrix $\prod_{i=1}^k\Tilde{W}_i$ is an $(\eps + O(\eps^2))$-SV approximation of $\prod_{i=1}^kW_i$. 
Moreover, SV approximation is preserved under the following operations that allow for efficient constructions.

\begin{lemma}
\label{fact:SV-properties}
SV approximation is preserved under the following operations:
    \begin{enumerate}
        \item (Decomposition:) If $\Tilde{A}_1$ and $\Tilde{A}_2$ are $\eps$-SV-approximations of matices $A_1$ and $A_2$ respectively, then $\Tilde{A}_1  + \Tilde{A}_2$ is an $\eps$-SV approximation of $A_1 + A_2$.
        \item (Bipartite Lift:) $\Tilde{A}$ is an $\eps$-SV approximation of $A$ if and only if the symmetric matrix 
        \[\begin{bmatrix}
            0&\Tilde{A}\\
            \Tilde{A}^\top &0
        \end{bmatrix} \text{ is an $\eps$-SV approximation of }
        \begin{bmatrix}
            0&A\\
            A^\top &0
        \end{bmatrix}. \]
    \end{enumerate}
\end{lemma}

The first property says that SV approximation is linear, which is crucial for decomposition-based construction algorithms. 
The second property says that to construct an SV sparsifier of a directed graph $G = (V, \vec{E})$, it is equivalent to construct an SV sparisifier of its bipartite lift, which is an undirected bipartite graph. 
Thus, unlike for standard or UC approximation, SV sparsification of directed graphs can be reduced to SV sparsification of undirected bipartite graphs. 
This will allow us to apply the partial coloring algorithm in \autoref{lemma:deterministic-partial-coloring} to construct SV sparsifiers for both directed and undirected graphs.


\section{Deterministic Matrix Partial Coloring} \label{s:partial-coloring}

The goal of this section is to prove \autoref{lemma:deterministic-partial-coloring}, with a proof overview given in \autoref{s:overview}.
We will first present the deterministic discrepancy walk framework used in previous work in \autoref{s:general-framework}.
Then, we will describe the potential function and use it to bound the maximum eigenvalue in \autoref{s:potential}.
Finally, we will describe the restricted subspace in \autoref{s:subspace} and bound the second-order term using a spectral argument in \autoref{s:trace} to complete the proof.

\subsection{The Deterministic Discrepancy Walk Framework} \label{s:general-framework}

We use the deterministic discrepancy walk framework in~\cite{LRR17,BLV22,PV23} for computing a partial coloring with small discrepancy.
The idea is to use a potential function that bounds the discrepancy of our current solution to guide a deterministic walk in a suitable high dimensional subspace.

 \begin{framed}{\textbf{Deterministic Discrepancy Walk for Matrix Partial Coloring}}
 
 \textbf{Input:} $A_1, \cdots, A_m \in \R^{n \times n}$ such that $\sum_{i=1}^m |A_i| \preccurlyeq I_n$ and a potential function $\Phi: \R^m \mapsto \R_+$

 \textbf{Output:} $x \in [-1, 1]^m$ such that $x$ has small discrepancy $\norm{A(x)}_{\rm op}$ and many coordinates in $\{\pm1\}$.
 \begin{enumerate}
     \item Initialization: Set $x_0 = 0_m$ as the initial point.  
Set $H_1 = [m]$ be the initial set of active coordinates. 
Set $t = 1$. Let $\alpha = 1/\poly(m)$ be the maximal step size. 
     \item \textbf{While} $m_t := |H_t| > \frac34 m$ \textbf{do}
         \begin{enumerate}
           \item  Pick $y_t$ to be a unit vector from an appropriate subspace $\mathcal{U}_t \subseteq \R^m$ satisfying $y_t\perp x_{t-1}$, $\supp(y_t)\subseteq H_t$, and $\Phi(x_{t-1} + y_t) - \Phi(x_{t-1})$ is bounded.
           \item Let $\delta_t$ be the largest step size such that $\delta_t \leq \alpha$ and $x_{t-1} + \delta_t y_t \in [-1,1]^m$. 
            \item Update $x_t \gets x_{t-1} + \delta_t y_t$. Update $t \gets t+1$.
            \item Update $H_{t} := \{ i \in [n] \mid |x_{t-1}(i)| < 1 \}$ be the set of fractional coordinates. 
       \end{enumerate}
   
    \item Return $x:=x_T$ where $T$ is the last iteration. 
 \end{enumerate}
 \end{framed}

To bound the discrepancy of the final solution $x_T$, we use a potential function $\Phi$ such that $\Phi(x)$ is a smooth estimation of the discrepancy of the solution $x$. 
We then bound the final discrepancy by $\Phi(x_T) = \Phi(0) + \sum_{t=1}^T \big( \Phi(x_{t}) - \Phi(x_{t-1}) \big)$. 
So, if $\Phi(0)$ is small and $\Phi(x_{t}) - \Phi(x_{t-1})$ is bounded for all $t$, then we can upper bound the discrepancy of $x_T$. 
The main task of the analysis under this framework then is to show that for a suitable potential function, and any $x_{t-1} \in [-1,1]^m$, there is a large enough subspace $\mathcal{U}_t$ such that for any unit vector $y_t \in \mathcal{U}_t$, the potential increase $\Phi(x_{t-1} + \delta_ty_t) - \Phi(x_{t-1})$ is small if the step size $\delta_t$ is at most $\alpha$. 
This implies that if the dimension of the allowed subspace $\mathcal H$ is also large enough, then there is always an update direction $y_t \in \mathcal{U}_t \cap \mathcal H$, which would imply that the final solution has small discrepancy and is in the allowed subspace $\mathcal H$ as well.

\subsection{Potential Function and Maximum Eigenvalue Bound} \label{s:potential}

Allen-Zhu, Liao, Orecchia~\cite{AZLO15} developed a regularized optimization framework to derive the spectral sparsification result in~\cite{BSS12} in a more principled way.
Recall that the maximum eigenvalue $\lambda_{\max}(A)$ of a matrix $A$ can be formulated as $\max_{M \in \Delta_n} \inner{A}{M}$ where $\Delta_n$ is the set of density matrices.
In this framework, the potential function is a regularized version of the maximum eigenvalue 
\[
\Phi(x) :=\max_{M \in \Delta_n} \inner{A(x)}{M} - \frac{1}{\eta} \cdot \phi(M),
\]
where the regularizer $\phi(M)$ is a nonpositive strongly convex function
and $\eta$ is a parameter controlling the contribution of the regularization term. 
They showed that the negative entropy regularizer $\phi(M) := \inner{M}{\log M}$ can be used to obtain a deterministic algorithm to recover the $O( (n \log n) / \eps^2)$ spectral sparsification result in~\cite{SS11},
and the $\ell_{1/2}$-regularizer $\phi(M) = -2\tr(M^{\frac12})$ can be used to recover the $O(n / \eps^2)$ spectral sparsification result in~\cite{BSS12}.

\subsubsection*{Potential Function}

Naturally, we set our potential function to be
\begin{equation} \label{eq:potential}
\Phi(x) := \max_{M \in \Delta_n} \inner{A(x)}{M} + \frac{2}{\eta} \cdot \tr(M^\frac12).
\end{equation}
It follows from \autoref{l:root-density} that $\Phi(x)$ satisfies 
$\lambda_{\max}(A(x)) \leq \Phi(x) \leq \lambda_{\max}(A(x)) + 2\sqrt{n}/\eta$.
For the final partial coloring $x_T = \sum_{t=1}^T \delta_t y_t$ returned by the general framework,
\begin{align*}
\lambda_{\max}(A(x_T)) \leq \Phi(x_T) 
& = \Phi(0) + \sum_{t=1}^T \big( \Phi(x_{t-1} + \delta_t y_t) - \Phi(x_{t-1}) \big)
\leq  \frac{2\sqrt{n}}{\eta} + \sum_{t=1}^T \big( \Phi(x_{t-1} + \delta_t y_t) - \Phi(x_{t-1}) \big).
\end{align*}

\subsubsection*{Maximum Eigenvalue and Operator Norm}

The potential function is used to upper bound $\lambda_{\max}(A(x))$, 
but the goal in \autoref{lemma:deterministic-partial-coloring} is to upper bound $\|A(x)\|_{\rm op}$. 
A standard trick in discrepancy theory is to handle the operator norm via the reduction
\[
\norm{A(x)}_{\rm op} = \lambda_{\max} \left( \sum_{i=1}^m  x(i) \cdot \begin{pmatrix} A_i & 0 \\ 0 & -A_i \end{pmatrix} \right).
\]
This only increases the dimension of the matrices by a factor of two.
We will abuse notation to still use $A_i$ to denote the block diagonal matrix $\begin{pmatrix} A_i & 0 \\ 0 & -A_i \end{pmatrix}$ throughout this section.

We remark that this is an advantage of the framework by Reis and Rothvoss~\cite{RR20} that allows the input matrices $A_1, \ldots, A_m$ to be arbitrary symmetric matrices (not just for rank one symmetric matrices of the form $v v^\top$) so that the above reduction works.
In previous work~\cite{BSS12,AZLO15} for spectral sparsification, two potential functions are used to keep track of the maximum eigenvalue and the minimum eigenvalue separately, and thus the algorithm and the analysis are more involved.
This is one place where the discrepancy walk framework provides simplification even for the standard setting in~\cite{BSS12,AZLO15}, as keeping track of only one potential function makes the algorithm and the analysis conceptually and technically simpler.

\subsubsection*{Maximum Eigenvalue Bound}

The next task is to bound the potential increase.
The following bound essentially\footnote{In Theorem 3.3 of~\cite{AZLO15}, the matrices are assumed to be either positive semidefinite or negative semidefinite, while our matrices $\begin{pmatrix} A_i & 0 \\ 0 & -A_i \end{pmatrix}$ do not satisfy this assumption.  But their arguments may be adapted to give the same result.}
follows from~\cite{AZLO15}.

\begin{lemma}[Potential Increase] \label{l:potential}
Given symmetric matrices $A_1, \cdots, A_m \in \R^{n \times n}$ and $x \in \R^m$, the unique optimizer in~(\ref{eq:potential}) is $M = (u I_n - \eta A(x))^{-2}$ where $u \in \R$ is the unique value such that $M \in \Delta_n$.
For any $y \in \R^m$, the change of potential function can be bounded by
\[
\Phi(x + y) - \Phi(x) 
\leq \tr(M A(y)) + c \eta \tr\big(M^{\frac12} A(y) M^{\frac12} A(y) M^{\frac12}\big),
\]
where $|c| \leq 2$ as long as $\norm{M^{\frac12} \cdot \eta A(y)}_{\rm op} \leq \frac12$.
\end{lemma}

The proof in~\cite{AZLO15} used some advanced concepts in optimization such as mirror descent and Bregman divergence.
We present a shorter and simpler proof in~\autoref{a:omitted}, which is essentially approximating the potential function using the second-order Taylor expansion by some elementary convexity arguments.

With \autoref{l:potential},
we set $x = x_{t-1}$, $y = \delta_t y_t$, and $M = (u_t I_n - \eta A(x_{t-1}))^{-2}$ for each iteration $t$ to bound
\begin{align} \label{eq:regret}
  \lambda_{\max}(A(x_T)) & \leq \frac{2\sqrt{n}}{\eta} + \sum_{t=1}^T \left( \tr\big(M_t A(\delta_t y_t)\big) + c_t \eta \tr\Big(M_t^{\frac12} A(\delta_t y_t) M_t^{\frac12} A(\delta_t y_t) M_t^{\frac12}\Big) \right).
\end{align}
Note that we can set the maximal step size of $\delta_t$ for all $t$ to be $\alpha := \frac{1}{2\eta}$ to ensure that\footnote{To look ahead, we will set $\eta \approx \sqrt{m}$ so that $\alpha \leq 1/\poly(m)$.} 
\begin{equation} \label{eq:step}
\norm{M^{\frac12} \cdot \eta A(\delta_t y_t)}_{\rm op} 
\leq \eta  \norm{\sum\nolimits_{i=1}^m \delta_t \cdot y_t(i) \cdot A_i}_{\rm op} \leq \eta \cdot \delta_t \cdot \norm{y_t}_{\infty} \cdot \norm{\sum\nolimits_{i=1}^m |A_i|}_{\rm op} \leq \eta \cdot \delta_t \leq \frac12,
\end{equation}
where we used the fact that $M$ is a density matrix, $y_t$ is a unit vector and the assumption in \autoref{lemma:deterministic-partial-coloring} that $\sum_{i=1}^m |A_i| \preccurlyeq I_n$.
Therefore, we can assume that $|c_t| \leq 2$ for all $t$ when we apply \autoref{l:potential}.

\subsection{Restricted Subspaces} \label{s:subspace}

The key task is to find an appropriate unit vector $y_t$ so that the potential increase in~\eqref{eq:regret} is bounded.
The essence of many results in discrepancy minimization is to argue that there are not too many bad directions, so that as long as the degree of freedom is large (that is, $m \gg n$) then there is a large subspace of good update directions.

To bound the potential increase in each iteration in~\eqref{eq:regret},
the nontrivial part is to bound the second order term.
We use the idea in~\cite{PV23} to write the second order term as the quadratic form of a matrix $N_t$ and restrict $y_t$ to lie in the small eigenspace of $N_t$.
Formally, let $H_t$ be the active coordinates in the $t$-th iteration and $m_t = |H_t|$.
The second order term in~\eqref{eq:regret} can be written as
\[
\tr\Big(M_t^{\frac12} A(y_t) M_t^{\frac12} A(y_t) M_t^{\frac12}\Big) 
= \sum_{i,j \in H_t} y_t(i) \cdot y_t(j) \cdot \tr\Big(M_t^{\frac12} A_i M_t^{\frac12} A_j M_t^{\frac12}\Big) 
= (y_t|_{H_t})^\top N_t (y_t|_{H_t}),
\]
where $y_t|_{H_t} \in \R^{m_t}$ is obtained by restricting $y_t$ to coordinates in $H_t$, and $N_t$ is an $(m_t \times m_t)$-dimensional matrix defined as
\[
N_t := \Big\{ \tr\Big(M_t^{\frac12} A_i M_t^{\frac12} A_j M_t^{\frac12} \Big) \Big\}_{i,j \in H_t}.
\]
Note that $N_t$ is a positive semidefinite matrix\footnote{
To see this, write $N_t(i,j) = \tr(M_t^{\frac12} A_i M_t^{\frac12} A_j M_t^{\frac12})
= \langle M_t^{\frac14} A_i M_t^{\frac12}, M_t^{\frac14} A_j M_t^{\frac12} \rangle $, and so $N_t$ can be written as $X^\top X$ where the $i$-th column of $X$ is ${\rm vec}( M_t^{\frac14} A_i M_t^{\frac12} )$, the vectorization of the matrix $M_t^{\frac14} A_i M_t^{\frac12}$.
}.
Let $N_t = \sum_{i=1}^{m_t} \lambda_i u_i u_i^\top$
be the eigenvalue decomposition of $N_t$ with $0 \leq \lambda_1 \leq \lambda_2 \leq \cdots \leq \lambda_{m_t}$. 
In order to bound the second order term, we restrict $y_t|_{H_t}$ to lie in the eigenspace spanned by $\{u_1, \ldots, u_{m_t/3}\}$. 
With these notations set up, we can formally define the good subspace $\mathcal{U} := U^0 \cap U^1 \cap U^2 \cap U^3$ for the unit update vector $y_t$ where
\begin{align*}
    & U^0 = \{ y \in \R^m \mid y_i = 0~\forall i \not\in H_t\}, \\
    & U^1 = \{ y \in \R^m \mid y \perp x_{t-1} \}, \\
    & U^2 = \left\{ y \in \R^m \mid \tr(M_t A( y)) = \sum\nolimits_{i=1}^m  y(i) \tr(M_t A_i) = 0 \right\}, \\
    & U^3 = \left\{ y \in \R^m \mid y|_{H_t} \in {\rm span}\{u_1, u_2, \cdots, u_{m_t/3}\} \right\}.
\end{align*}
As discussed in the overview in~\autoref{s:overview},
the subspace $U^0$ is to ensure that only active coordinates are updated,
$U^1$ is to ensure that $\|x_t\|^2_2$ is monotone increasing in order to bound the number of iterations of the algorithm, 
$U^2$ is to ensure that the linear term in~\eqref{eq:regret} is zero, 
and $U^3$ is to ensure that the second order term in~\eqref{eq:regret} is ``small''.
As $m_t \geq \frac34 m$ when the algorithm has not terminated, it follows that
\begin{equation} \label{eq:dim}
{\rm dim}(\mathcal{U}) \geq m_t - 2 - \frac{2m_t}{3} = \frac{m_t}{3} - 2 \geq \frac{m}{4} -2.
\end{equation}

The remaining task is to upper bound the eigenvalue of the low eigenspace in order to upper bound the second order term.

\begin{lemma}[Low Eigenspace] \label{l:quadratic}
    Given $A_1, \cdots, A_m \in \R^{n \times n}$ such that $\sum_{i=1}^m |A_i| \preccurlyeq I_n$, any unit vector $y \in U^0 \cap U^3$ satisfies
    \[
        \tr\Big(M^{\frac12}_t A(y) M_t^{\frac12} A(y) M_t^{\frac12}\Big) \leq \frac{9\sqrt{n}}{m_t^2}.
    \]
\end{lemma}

We will use a spectral argument to prove \autoref{l:quadratic} in the next subsection.
In the rest of this subsection, we first assume \autoref{l:quadratic} to finish the proof of \autoref{lemma:deterministic-partial-coloring}.

\subsubsection*{Proof of \autoref{lemma:deterministic-partial-coloring}}

Given the input matrices $A_1, \cdots, A_m$ such that $\sum_{i=1}^m |A_i| \preccurlyeq I_n$ and the linear subspace $\mathcal{H}$, 
we apply the deterministic discrepancy walk algorithm for matrix partial coloring with the subspace $\mathcal{U}_t = \mathcal{U} \cap \mathcal{H}$. 
By \eqref{eq:dim},
\[{\rm dim}(\mathcal{U}_t) 
\geq {\rm dim}(\mathcal{U}) - {\rm dim}(\mathcal{H}^\perp) 
\geq \frac{m}{4} - 2 - \frac{m}{5} > 0\]
as long as $m = \Omega(1)$, 
thus in each iteration there is always a unit vector $y_t \in \mathcal{U}_t$.
As shown in~\eqref{eq:step}, by taking the maximal step size to be $\alpha = \frac{1}{2\eta}$, we can assume that $|c_t| \leq 2$ in~\eqref{eq:regret}, and so
\[
\lambda_{\max}\left( A(x_T) \right)  
\leq \frac{2\sqrt{n}}{\eta} + \sum_{t=1}^T \left( \delta_t \tr(M_t A(y_t)) + 2 \eta \delta_t^2 \tr\Big(M_t^{\frac12} A(y_t) M_t^{\frac12} A(y_t) M_t^{\frac12}\Big) \right).
\]
Since $y_t \in U^2$, the linear term is $\tr(M_t A(y_t)) = 0$. 
As $y_t \in U^0 \cap U^3$, by \autoref{l:quadratic}, the second order term is $\tr(M^{\frac12}_t A(y_t) M_t^{\frac12} A(y_t) M_t^{\frac12}) \leq 9\sqrt{n}/m_t^2$. 
Therefore,
\[
\lambda_{\max}\left( A(x_T) \right)  
\leq \frac{2\sqrt{n}}{\eta} + \frac{18 \eta \sqrt{n}}{m_t^2} \sum_{t=1}^T \delta_t^2 \leq \frac{2\sqrt{n}}{\eta} + \frac{32 \eta \sqrt{n}}{m^2} \sum_{t=1}^T \delta_t^2,
\]
where the last inequality holds as $m_t \geq 3m/4$ before the while loop terminates.
Since $y_t \in U^1$ (such that $y_t \perp x_{t-1}$ for all $t \in [T]$) and $\norm{x_T} \in [-1,1]^m$, it follows that
$m \geq \norm{x_T}_2^2 = \sum_{t=1}^T \norm{\delta_t y_t}_2^2 = \sum_{t=1}^T \delta_t^2$.
Therefore, by setting $\eta = \frac14 \sqrt{m}$, we conclude that
\[
\lambda_{\max}\left( A(x_T) \right)  \leq \frac{2\sqrt{n}}{\eta} + \frac{32 \eta \sqrt{n}}{m} \leq 16 \sqrt{\frac{n}{m}}.
\]
    
Finally, to show the polynomial runtime, we bound the number of iterations of the deterministic walk algorithm. 
Note that in step 2(c), the update either (i) freezes a new coordinate, 
or (ii) the squared length of the solution $\norm{x_t}^2 = \norm{x_{t-1}}^2 + \alpha^2$ increases by $\alpha^2$ as $y_t\perp x_{t-1}$.
Clearly, the number of the first type of iterations is at most $m$. 
The number of the second type of iterations is at most $m/\alpha^2$ as $\norm{x_T}^2\leq m$. 
Therefore, by our choice of $\alpha = \frac{1}{2\eta} = \frac{2}{\sqrt{m}}$, the total number of iterations is at most $m/\alpha^2 + m = O(m^2)$.

\subsection{Spectral Argument} \label{s:trace}

We prove \autoref{l:quadratic} in this subsection using a spectral argument.
We remark that the calculations are similar to that in~\cite[page 13]{RR20}, but we adapt them to a slightly different setting.

For a unit vector $y \in U^0$, the restricted vector $y|_{H_t}$ is a unit vector in $\R^{H_t}$.
Since $y \in U^3$, the second order term is bounded by the eigenvalue of $N_t$ in the low eigenspace so that
\[
\tr \Big( M^{\frac12}_t A(y) M_t^{\frac12} A(y) M_t^{\frac12} \Big) 
= (y|_{H_t})^\top N_t (y|_{H_t}) \leq \lambda_{m_t/3}(N_t).
\]
To bound $\lambda_{m_t/3}(N_t)$, our idea is to upper bound the trace of a large principal submatrix $\Tilde{N_t}$ of $N_t$ and to use Cauchy interlacing theorem to bound $\lambda_{m_t / 3}$.
Let
\[
S = \Bigg\{ i \in H_t~\Big|~\tr(M_t^{\frac12} |A_i|) \geq \frac{3 \tr(M_t^{\frac12})}{m_t} \Bigg\}
\]
be the set of ``large'' active coordinates.  
Note that $\sum_{i \in H_t} \tr(M_t^{\frac12} |A_i|) \leq \tr(M_t^{\frac12}) \cdot \norm{\sum_{i=1}^m |A_i|}_{\rm op} \leq \tr(M_t^{\frac12})$ (see \autoref{s:linear-algebra} for the first inequality) and each $\tr(M_t^{\frac12} |A_i|) \geq 0$, so it follows from Markov's inequality that $|S| \leq \frac13 m_t$.
Let $\tilde{N}_t$ be the principal submatrix of $N$ restricted to the indices in $H_t - S$. 
Note that ${\rm dim}(\tilde{N}_t) \geq m_t - |S| \geq \frac23 m_t$. 
By Cauchy interlacing in \autoref{t:interlacing}, it holds that
$\lambda_{m_t/3}(N_t) \leq \lambda_{m_t/3}(\tilde{N}_t)$. 
To bound $\lambda_{m_t/3}(\tilde{N}_t)$, we will simply compute the trace of $\tilde{N}_t$ and use an averaging argument. 
Applying \autoref{l:RR} to each diagonal entry of $\tilde{N}_t$, we have
\[
 \tr(\tilde{N}_t)  = \sum_{i \in H_t - S} \tr(M_t A_i M_t^{\frac12} A_i) \leq \sum_{i \in H_t - S} \tr(M_t |A_i|) \cdot \tr\big(M_t^{\frac12} |A_i|\big).
\]
Now, by the definition of $S$ and the assumption that $\sum_{i=1}^m |A_i| \preccurlyeq I_n$ and \autoref{l:root-density}, we obtain that
\[
    \tr(\tilde{N}_t) \leq \frac{3\tr(M_t^{1/2})}{m_t} \cdot \sum_{i \in H_t - S} \tr(M_t |A_i|)  \leq \frac{3\tr(M_t^{1/2})}{m_t} \cdot \tr(M_t) \leq \frac{3\sqrt{n}}{m_t}.
\]
Since ${\rm dim}(\tilde{N}_t) \geq \frac23 m_t$, the average value of the eigenvalue of $\tilde{N}_t$ is $\tr(\tilde{N}_t) / {\rm dim}(\tilde{N}_t) \leq \frac92 \sqrt{n}/m_t^2$.
As $\tilde{N}_t$ is positive semidefinite, by Markov's inequality,
at most half of the eigenvalues of $\tilde{N}_t$ can be greater than $9 \sqrt{n}/m_t^2$.
Combining the inequalities, we conclude that 
\[
\tr \Big( M^{\frac12}_t A(y) M_t^{\frac12} A(y) M_t^{\frac12} \Big) 
\leq \lambda_{m_t/3}(N_t)
\leq \lambda_{m_t/3}(\tilde{N}_t) 
\leq \lambda_{{\rm dim}(\tilde{N}_t)/2}(\tilde{N}_t) 
\leq \frac{9 \sqrt{n}}{m_t^2}.
\]

\section{From Matrix Partial Coloring to Spectral Sparsification} \label{s:SV-sparsification}

In this section, we show how to apply the matrix partial coloring \autoref{lemma:deterministic-partial-coloring} to construct spectral sparsifers.
In \autoref{s:spectral-subspace}, we prove \autoref{theorem:sparsification-alg}  using the same reduction as in~\cite{RR20}, and obtain deterministic degree-preserving spectral sparsification as a corollary.
In \autoref{ss:UC-sparsification} and \autoref{ss:SV-sparsification}, we use \autoref{theorem:sparsification-alg} to construct UC sparsifiers as stated in \autoref{theorem: undirected-uc-sparsifier} and SV sparsifiers as stated in \autoref{theorem: SV-sparsfication-general}.

\subsection{Spectral Sparsification with Linear Subspace Constraints} \label{s:spectral-subspace}

The goal of this subsection is to prove \autoref{theorem:sparsification-alg}, a general matrix sparsification result with a linear subspace constraint.
We emphasize that the reduction from matrix sparsification in \autoref{theorem:sparsification-alg} to matrix partial coloring in \autoref{lemma:deterministic-partial-coloring} is exactly the same as in~\cite{RR20}, both the algorithm and the analysis.
We include them for the completeness of this paper and for the inclusion of the linear subspace constraint.

\begin{framed}{\textbf{Spectral Sparsification with Linear Subspace Constraint}} 

\textbf{Input:} positive semidefinite matrices $A_1, A_2, ..., A_m \in \R^{n \times n}$ such that $\sum_{i=1}^m A_i \preccurlyeq I_n$, a linear subspace $\mathcal{H} \subseteq \R^m$ of dimension $m-O(n)$, and target accuracy parameter $\eps$.

\textbf{Output:} $s \in \R^m_{\geq 0}$ such that $|\supp(s)| = \Omega(n/\eps^2)$ and $s - \vec{\one} \in \mathcal{H}$ and $\|\sum_{i=1}^m (s(i)-1) A_i\|_{\rm op} \leq O(\eps)$.

    \begin{enumerate}
        \item Initialize $s_0(i) = 1$ for $i \in [m]$. Let $t = 1$.
        \item \textbf{While} $m_t := |\supp(s_{t-1})| > cn/\eps^2$ for some fixed constant $c$ {\bf do}
        \begin{enumerate}
            \item Apply the deterministic discrepancy walk algorithm in~\autoref{lemma:deterministic-partial-coloring} to find a partial coloring $x_t \in [-1,1]^m$ such that
            \begin{enumerate}
                \item $\norm{\sum_i x_t(i) \cdot s_{t-1}(i) \cdot A_i}_{\rm op} \leq 16\sqrt{\frac{n}{m_t}}$,
                \item $|\{i \in \supp(s_{t-1}) \mid x_t(i) = \pm 1\}| = \Omega(m_t)$ and $x_t(i) = 0$ for all $i \not\in \supp(s_{t-1})$, 
                \item $x_t \in \mathcal{H}_t := \{y \in \R^m \mid \diag(s_{t-1}) \cdot y \in \mathcal{H} \} \cap \{ y \in \R^m \mid y(i) = 0~\forall i \notin \supp(s_{t-1})\}$.
            \end{enumerate}
            \item If there are more $x_t(i) = 1$ than $x_t(i) = -1$ then update $x_t \gets - x_t$.
            \item Update $s_{t}(i) \gets s_{t-1}(i) \cdot (1+x_t(i))$ for all $i \in [m]$.
            \item $t\gets t+1$.
        \end{enumerate}
        \item \textbf{Return:} $s = s_T$ where $T$ is the last iteration.
    \end{enumerate}
\end{framed}

The idea is to find a partial coloring $x$ with small discrepancy (with a constant fraction of entries of $x$ being $-1$) and use it to zero-out a constant fraction of the entries in $s$ in each iteration.
Informally, if $x \in \{\pm1\}^m$ is a full coloring with small discrepancy, then we either double-up or zero-out an entry of $s$ in each iteration.

\subsubsection*{Proof of \autoref{theorem:sparsification-alg}}

We start by assuming that the spectral sparsification algorithm can successfully find a partial coloring $x_t$ that meets all the requirements (i)-(iii) in Step 2(a) in each iteration.
It is immediate that the final reweighting will have support size at most $O(n/\eps^2)$ by the design of the algorithm.
Next, the final reweighting $s$ would satisfy the subspace requirement $s - \vec{\one} \in \mathcal{H}$ as 
\[
s_t = s_{t-1} + \diag(s_{t-1}) \cdot x_t {\rm~~for~all~} t \quad {\rm~and~} \quad s_0 = \vec{\one} 
\quad \implies \quad
s = \vec{\one} + \sum\nolimits_{t=1}^T \diag(s_{t-1}) \cdot x_t
\]
and each $\diag(s_{t-1}) \cdot x_t \in \mathcal{H}$ by the requirement (iii) in Step 2(a) in each iteration. 
For the discrepancy guarantee, first note that Step 2(a)(ii) and Step 2(b) imply that $x_t$ has $\Omega(m_t)$ coordinates with $-1$ in the support of $s_{t-1}$, and thus the support size $m_t$ decreases by a constant factor in each iteration by Step 2(c) of the algorithm.
Therefore, by a telescoping sum and the triangle inequality, the discrepancy is
\begin{equation} \label{eq:discrepancy}
\norm{A(s) - A(\vec{\one}_m)}_{\rm op} 
\leq \sum_{t=1}^T \norm{\sum_{i=1}^m \big(s_t(i) - s_{t-1}(i)\big) \cdot A_i }_{\rm op} 
= \sum_{t=1}^T \norm{\sum_{i=1}^m x_t(i) \cdot s_{t-1}(i) \cdot A_i }_{\rm op} 
\lesssim \sum_{t=1}^T \sqrt{\frac{n}{m_t}} 
\lesssim \eps,
\end{equation}
where the second last inequality follows by requirement (i) in Step 2(a) in each iteration, and the last inequality follows as $m_t$ is a geometric sequence dominated by the last term when $m_t > cn/\eps^2$.
Finally, for the polynomial running time, as the support size $m_t$ decreases by a constant factor in each iteration, the while loop will terminate within $O(\log(\eps^2 m/n))$ iterations.

It remains to show that we can indeed find the desired partial coloring $x_t$ in each iteration. 
We maintain that $\sum_{i=1}^m s_{t-1}(i) \cdot A_i \preccurlyeq 2 I_n$ at each iteration $t$, which obviously holds in the first iteration as $s_0 = \vec{\one}$. 
This implies that $\sum_{i=1}^m \frac12 s_{t-1}(i) \cdot A_i\preccurlyeq I_n$ satisfies the requirement of \autoref{lemma:deterministic-partial-coloring}.
Note that the input subspace $\mathcal{H}_t$ has dimension $m_t - O(n) \geq \frac45 m_t$ for $m_t = \Omega(n / \eps^2)$ when $\eps$ is smaller than a small enough constant.
So we can apply \autoref{lemma:deterministic-partial-coloring} with $\{\frac12 s_{t-1}(i) A_i\}_{i=1}^m$ as the input matrices and $\mathcal{H}_t$ as the input subspace to find a partial coloring $x_t$ that satisfies all the requirements in the matrix sparsification algorithm in polynomial time.
The property $\sum_{i=1}^m s_{t-1}(i) \cdot A_i \preccurlyeq 2 I_n$ is maintained by the same argument in~\eqref{eq:discrepancy} as long as $\eps$ is smaller than a small enough constant.
This completes the proof of \autoref{theorem:sparsification-alg}.

\begin{remark}
Note that the input matrices in \autoref{theorem:sparsification-alg} are assumed to be positive semidefinite (as was done in~\cite{RR20}) while the input matrices in \autoref{lemma:deterministic-partial-coloring} are only assumed to be symmetric.
We can relax the assumption in \autoref{theorem:sparsification-alg} to be symmetric as well, by the simple trick of replacing the input symmetric matrices $A_i$ by $\begin{pmatrix} A_i & 0 \\ 0 & |A_i| \end{pmatrix}$, so that the property $\sum_{i=1}^m s_{t-1}(i) \cdot |A_i| \preceq 2I_n$ can be maintained by the discrepancy bound and the above proof would go through.
\end{remark}

\subsubsection*{Degree-Preserving Spectral Sparsification}

We show an easy application of \autoref{theorem:sparsification-alg} to construct degree-preserving spectral sparsifiers in deterministic polynomial time, derandomizing a result in~\cite{JRT24}.
Given an edge-weighted undirected graph $G=(V,E)$,
construct a positive semidefinite matrix $A_i = L^{\dagger/2}_G b_{e_i} b_{e_i}^\top L_G^{\dagger/2}$ for each edge $e_i \in E$,
so that $\sum_{i=1}^{m} A_i = I_n$ where $n = |V|$ and $m = |E|$.
For the degree constraints, construct the subspace $\mathcal H := \{x \mid \sum_{u: u \sim v} x(uv)=0~\forall v \in V \}$ to ensure that $\sum_{u:u \sim v} s(uv) \cdot w(uv) = \sum_{u:u \sim v} w(uv)$ for each vertex $v \in V$.
Note that $\dim({\mathcal H}) \geq m - n$.
Therefore, we can apply \autoref{theorem:sparsification-alg} with $A_1, \ldots, A_m$ and $\mathcal H$ to obtain a deterministic polynomial time algorithm to construct linear-sized degree-preserving spectral sparsifiers.

\subsection{Unit-Circle Sparsification for Undirected Graphs} \label{ss:UC-sparsification}

The goal of this subsection is to prove \autoref{theorem: undirected-uc-sparsifier} that there are linear-sized UC sparsifiers for undirected graphs.
This is a nice application to illustrate our deterministic discrepancy walk framework, as it follows very easily from our matrix sparsification result,
 while it seems out of reach for the known techniques in~\cite{SS11,BSS12,AZLO15}.
Recall from \autoref{s:preliminaries} that UC approximation enjoys some interesting property that was not satisfied by standard spectral sparsification, 
so \autoref{theorem: undirected-uc-sparsifier} can be seen as a strengthening of the classical results in~\cite{BSS12,AZLO15}.

We mentioned in the introduction that $\Tilde{G}$ is an $\eps$-UC approximation of $G$ if (1) $D_{\Tilde{G}} = D_G$, (2) $(1-\eps)L_G \preceq L_{\Tilde{G}} \preceq (1+\eps)L_G$, and (3) $(1-\eps)U_G \preceq U_{\Tilde{G}} \preceq (1+\eps)U_G$ where $U_G := D_G + A_G$ which is often called the unsigned Laplacian matrix of $G$.
We check this claim formally in the following lemma.

\begin{lemma}\label{lemma: undirected-uc-sparsifier}
Let $\Tilde{G}$ be a degree-preserving reweighted subgraph of a connected undirected graph $G$.
Then $\Tilde{G}$ is an $\eps$-UC approximation of $G$ if
\[
\norm{L_G^{\dagger/2}(L_G-L_{\Tilde{G}})L_G^{\dagger/2}}_{\rm op} \leq \eps
\quad {\rm and} \quad
\norm{U_G^{\dagger/2}(U_G-U_{\Tilde{G}})U_G^{\dagger/2}}_{\rm op} \leq \eps.
\]
\end{lemma}
\begin{proof}
Let the adjacency matrices of $G$ and $\Tilde{G}$ be $A$ and $\Tilde{A}$ respectively. 
By \autoref{def: uc-approx}, 
$\Tilde{G}$ is an $\eps$-UC approximation of $G$ if and only if (i) $\Tilde{A}$ is an $\eps$-standard approximation of $A$ and (ii) $-\Tilde{A}$ is an $\eps$-standard approximation of $-A$.
For undirected graphs, by \autoref{def:standard-approx}, (i) means that $\Tilde{A}$ is an $\eps$-matrix approximation of $A$ with respect to the error matrix $D-A = L$, and (ii) means that $-\Tilde{A}$ is an $\eps$-matrix approximation of $-A$ with respect to the error matrix $D+A = U$.

To check (i), by \autoref{def:matrix-approx}, it is equivalent to $\big\| L_G^{\dagger/2} (A_G-A_{\Tilde{G}})L_G^{\dagger/2} \big\|_{\rm op} \leq \eps$ and $\ker(L_G) \subseteq \ker(A_G-A_{\Tilde{G}})$.
Since $D_G = D_{\Tilde{G}}$ as $\Tilde{G}$ is degree-preserving, the first condition becomes $\big\| L_G^{\dagger/2}(L_G-L_{\Tilde{G}})L_G^{\dagger/2} \big\|_{\rm op} \leq \eps$, which is satisfied by the assumption.
As $G$ is connected and so $\ker(L_G) = \vec{\one}$, the second condition reduces to checking $(A_G - A_{\Tilde{G}}) \vec{\one} = 0$, which is also satisfied as $\Tilde{G}$ is degree-preserving.

To check (ii), by \autoref{def:matrix-approx}, it is equivalent to $\big\| U_G^{\dagger/2} (-A_G+A_{\Tilde{G}}) U_G^{\dagger/2} \big\|_{\rm op} \leq \eps$ and $\ker(U_G) \subseteq \ker(A_G-A_{\Tilde{G}})$.
Since $D_G = D_{\Tilde{G}}$ and $U_G = D_G + A_G$, the first condition is equivalent to $\big\| U_G^{\dagger/2}(U_G-U_{\Tilde{G}}) U_G^{\dagger/2} \big\|_{\rm op} \leq \eps$, which is satisfied by the assumption.
To check the second condition, we characterize the kernel of $U_G$.
The kernel of $U_G$ is non-empty if and only if $U_G$ has a zero eigenvalue.
Note that $U_G = 2D_G - L_G$, so it has a zero eigenvalue if and only if $D_G^{-\frac12} U_G D_G^{-\frac12} = 2I - \L_G$ has a zero eigenvalue if and only if $\L_G$ has an eigenvalue of $2$.
Hence, by \autoref{lemma: max-eigenvalue-bipartite}, $U_G$ has a zero eigenvalue if and only if $G$ is bipartite, and $\ker(U_G) = {\rm span}(\vec{\one}_X - \vec{\one}_Y)$ where $(X,Y)$ is the unique bipartition of $G$ as $G$ is connected.
Therefore, the second condition reduces to checking $(A_G - A_{\Tilde{G}})(\vec{\one}_X - \vec{\one}_Y) = 0$.
Since $\Tilde{G}$ is a subgraph of $G$ and thus $(X,Y)$ is also a bipartition of $\Tilde{G}$, it is straightforward to verify that $A_G(\vec{\one}_X - \vec{\one}_Y) = D_G( \vec{\one}_X - \vec{\one}_Y )$ and $A_{\Tilde{G}}(\vec{\one}_X - \vec{\one}_Y) = D_{\Tilde{G}}( \vec{\one}_X - \vec{\one}_Y )$, 
and thus the second condition is also satisfied because $D_G = D_{\Tilde{G}}$.
\end{proof}

\subsubsection*{Proof of \autoref{theorem: undirected-uc-sparsifier}}

With \autoref{lemma: undirected-uc-sparsifier}, it is easy to reduce UC approximation of undirected graphs to the matrix sparsification result in~\autoref{theorem:sparsification-alg}.
We already know how to do (1) and (2) simultaneously as it is just degree-preserving spectral sparsification.
Note that we also know to do (3), as $U_G = \sum_{e \in E} \bar{b}_e \bar{b}_e^\top$ where $\bar{b}_{uv} := \vec{\one}_u + \vec{\one}_v$ is the unsigned edge-incidence vector, and so we can write (3) as a matrix sparsification problem when the input matrices sum to the identity matrix.
To do (1), (2), and (3) simultaneously, we just need to do the standard block matrix trick, which is an advantage that the framework in~\cite{RR20} offered as it works for arbitrary rank symmetric matrices.

Given an undirected graph $G=(V,E)$,
for each edge $e \in E$, construct a positive semidefinite matrix 
\begin{align*}
    A_e = \begin{bmatrix}
        L_G^{\dagger/2}b_{e} b_{e}^\top L_G^{\dagger/2} & 0\\
        0 & U_G^{\dagger/2} \bar{b}_{e} \bar{b}_{e}^\top U_G^{\dagger/2}
     \end{bmatrix},
\end{align*}
so that $\sum_{e \in E} A_e = I_{2n}$ where $n = |V|$.
Also, construct the subspace $\mathcal H := \{x \mid \sum_{u: u \sim v} x(uv)=0~\forall v \in V \}$ to ensure degree preservation as in the previous subsection.
Applying \autoref{theorem:sparsification-alg} with $\{A_e\}_{e \in E}$ and $\mathcal H$ will give a deterministic polynomial time algorithm to construct a linear-sized degree-preserving sparsifier $\Tilde{G}$,
where the first block of the discrepancy guarantee implies that 
$\big\| L_G^{\dagger/2}(L_{\Tilde{G}} - L_G)L_G^{\dagger/2} \big\|_{\rm op} \leq \eps$ 
and the second block of the discrepancy guarantee implies that 
$\big\| U_G^{\dagger/2}(U_{\Tilde{G}} - U_G)U_G^{\dagger/2} \big\|_{\rm op} \leq \eps$.
We conclude from \autoref{lemma: undirected-uc-sparsifier} that $\Tilde{G}$ is a linear-sized $\eps$-UC sparisfier of $G$.

\subsection{Singular Value Sparsification for Directed Graphs} \label{ss:SV-sparsification}

The goal of this section is to prove \autoref{theorem: SV-sparsfication-general} about SV sparsifiers for directed graphs. 
By the second item in \autoref{fact:SV-properties}, the problem of constructing an SV sparsifier of a directed graph can be reduced to the problem of constructing an SV sparsifier of its bipartite lift which is an undirected graph. 
Henceforth, we will assume that $G$ is an undirected bipartite graph. 

By the definition of SV approximation in \autoref{def:SV-approx}, we would like to approximate the adjacency matrix $A$ with respect to the error matrix $D - AD^{-1}A$.
Unlike in the case of UC-approximation, we do not know how to reduce the general problem to matrix sparsification where the input matrices sum to the identity matrix.
However, we observe that if the bipartite graph is an expander graph, then we can reduce to matrix sparsification where the sum of input matrices has bounded spectral norm.
This observation leads to linear-sized SV sparsifier for bipartite $\Omega(1)$-expander.

\begin{theorem}[SV Approximation for Bipartite Expanders]
\label{theorem:SV-sparsification-expander}
Let $G$ be a connected bipartite graph with $\lambda_2(\L_G) \geq \lambda$. 
Then there is a polynomial time deterministic algorithm to compute a $(\eps / \lambda)$-SV sparsifier of $G$ with $O(n / \eps^2)$ edges.
\end{theorem}

We will use the following lemma for SV approximation which is an analog of \autoref{lemma: undirected-uc-sparsifier} for UC-sparsification.
The proof is to show that $\ker(E) = \text{span}\{\vec{\one}_V, \vec{\one}_X - \vec{\one}_Y\}$ where $(X,Y)$ is the unique bipartition of $G$, and the rest of the argument follows the same way as in the proof of \autoref{lemma: undirected-uc-sparsifier}.  We omit the straightforward proof. 

\begin{lemma}
\label{fact:SV-conditions}
Let $\Tilde{G}$ be a degree-preserving reweighted subgraph of a connected bipartite graph $G$.  Let $E = D_G-A_GD_G^{-1}A_G$. 
Then $\Tilde{G}$ is an $\eps$-SV approximation of $G$ if
\[\norm{E^{\dagger/2}(L_G- L_{\Tilde{G}})E^{\dagger/2}}_{\rm op} \leq \eps.\]
\end{lemma}

\subsubsection*{Proof of \autoref{theorem:SV-sparsification-expander}}

Given a bipartite graph $G=(V,E)$ with $\lambda_2(\L_G) \geq \lambda$,
for each edge $e \in E$, construct a positive semidefinite matrix
$A_e = \lambda E^{\dagger/2}b_eb_e^\top E^{\dagger/2}$.
We claim that $\sum_{e \in E} A_e \preceq I$.
Using the identity $E = (D-A)D^{-1}(D+A) = (D+A)D^{-1}(D-A)$,
we see that
\[
\sum\nolimits_{e} A_e 
= \lambda E^{\dagger/2}L_GE^{\dagger/2} 
= \lambda (D+A)^{\dagger/2}D(D+A)^{\dagger/2},
\] 
where we plug in $E^{\dagger/2} = (D+A)^{\dagger/2}D^{1/2}(D-A)^{\dagger/2}$ on the left and $E^{\dagger/2} = (D-A)^{\dagger/2}D^{1/2}(D+A)^{\dagger/2}$ on the right to obtain the second equality. 
Some simple manipulations show that
\[
\frac{1}{\lambda} \norm{\sum\nolimits_e A_e}
= \norm{(D+A)^{\dagger/2}D(D+A)^{\dagger/2}} 
= \norm{D^{\frac12}(D+A)^{\dagger}D^{\frac12}} 
= \norm{(D^{-\frac12}(D+A)D^{-\frac12})^{\dagger}}
= \norm{(2I-\L_G)^\dagger},
\]
where the second equality is by $\norm{BB^*}=\norm{B^*B}$, the third equality is by $\norm{B} = \norm{(B^\dagger)^\dagger}$, and the final equality is by $D^{-\frac12} A D^{-\frac12} = I - \L$.
Since the eigenvalues of $\L_G$ of a connected bipartite graph $G$ satisfy $0 = \lambda_1 < \lambda_2 \leq \ldots \leq \lambda_{n-1} < \lambda_n = 2$, 
the claim follows as
\[
\frac{1}{\lambda} \norm{\sum\nolimits_e A_e}
= \norm{(2I-\L_G)^\dagger}
= \frac{1}{2 - \lambda_{n-1}(\L_G)}
= \frac{1}{\lambda_2(\L_G)}
\leq \frac{1}{\lambda}
\quad \implies \quad
\sum\nolimits_e A_e \preceq I,
\]
where we used that \autoref{lemma: symmetric-spectrum} in the last equality 
and the implication is because $\sum_e A_e$ is a positive semidefinite matrix.

Applying \autoref{theorem:sparsification-alg} with $\{A_e\}_{e \in E}$ and the degree-constrained subspace $\mathcal H := \{x \mid \sum_{u: u \sim v} x(uv)=0~\forall v \in V\}$ will give a deterministic polynomial time algorithm to construct a linear-sized degree preserving sparsifier $\Tilde{G}$,
where the discrepancy guarantee implies that
$\lambda \big\| E^{\dagger/2}(L_{\Tilde{G}} - L_{G})E^{\dagger/2} \big\| \leq \eps$.
We conclude from \autoref{fact:SV-conditions} that $\Tilde{G}$ is a linear-sized $(\eps/\lambda)$-SV sparsifier of $G$. 

\begin{remark}
We note that some similar calculations were done in~\cite[Lemma 4.3]{APPSV23} and ours were inspired by theirs,
but we would like to point out that their purpose was different and in particular not for the construction of sparsification algorithms. 
\end{remark}

\subsubsection*{Proof of \autoref{theorem: SV-sparsfication-general}}

To derive \autoref{theorem: SV-sparsfication-general},
we will use expander decomposition as was done in \cite{APPSV23}. 
By \autoref{fact: expander-decomp}, every graph $G$ can be decomposed into $G_1, \ldots, G_k$ such that each $G_i$ is an $\Omega(1 / \log^2{n})$-expander and each vertex is contained in at most $O( \log n)$ subgraphs. 
For each $G_i$, we apply \autoref{theorem: SV-sparsfication-general} with error parameter $\eps' := \Omega(\eps / \log^2{n} )$ to find a graph $\Tilde{G}_i$ which is an $O(\eps)$-SV sparsifier of $G_i$, with $O( (n_i \log^4{n}) / \eps^2 )$ edges where $n_i = |V(G_i)|$. 
By the linearity property of SV approximation in \autoref{fact:SV-properties}, 
the graph $\Tilde{G} := \Tilde{G}_1 \cup...\cup \Tilde{G}_k$ is an $\eps$-SV sparsifier of $G$. 
Since each vertex in $G$ is contained in at most $O(\log{n})$ subgraphs, it follows that $\sum_{i=1}^k n_i = O(n \log n)$ and thus the total number of edges in $\Tilde{G}$ is $O( (n\log^5{n}) / \eps^2 )$.

\section{Graphical Spectral Sketch and Resistance Sparsifier} \label{s:sketch}

In this section, we show how to use the discrepancy framework to construct graphical spectral sketches in \autoref{s:graphical} and effective resistance sparsifiers in \autoref{s:resistance}.
We remark that there are some new technical ideas needed for proving \autoref{theorem: deterministic-eps-sketch} for spectral sketches, while \autoref{cor:resistance-sparsifier} for resistance sparsifiers follows easily from a reduction to graphical spectral sketch in~\cite{CGP+23} and \autoref{theorem: deterministic-eps-sketch}.
We assume that the input graphs are undirected and unweighted in this section.

\subsection{Graphical Spectral Sketch} \label{s:graphical}

The definition of spectral sketches in~\cite{ACK+16,JS18,CGP+23} is inherently probabilistic.
To design deterministic algorithms, we formulate the following deterministic version of $\eps$-spectral sketches.

\begin{definition}[Deterministic Graphic Spectral Sketch]
Given a graph $G = (V,E)$ on $n$ vertices and a set of vectors $\K \subseteq \R^n$, we say that a weighted graph $\Tilde{G}$ is an $\eps$-graphical spectral-sketch with respect to $\K$ if 
\begin{align*}
 (1 - \eps) z^\top L_G z \leq z^\top L_{\Tilde{G}} z \leq (1 + \eps) z^\top L_G z \quad {\rm for~all~} z \in \K.
\end{align*}
\end{definition}

\subsubsection*{Spectral Sketch from Partial Coloring}

We reduce sparsification to partial coloring as in \autoref{s:spectral-subspace}. 
In the following deterministic graphical sketch algorithm, the discrepancy requirement of the partial coloring in \eqref{eq:vector} is constructed so that the same argument as in \autoref{theorem:sparsification-alg} would go through to establish the spectral sketch guarantee.

\begin{framed}{\textbf{Deterministic Spectral Sketch Algorithm}} %

\textbf{Input:} an unweighted undirected $G = (V,E)$ on $n$ vertices, a set of vectors $\K \in \R^{n}$, and a target accuracy parameter $\eps$.

\textbf{Output:} a weighted undirected graph $\Tilde{G} = (V, \Tilde{E})$ with $|\Tilde{E}| \leq n \cdot f(n) / \eps$ and $z^\top L_{\Tilde{G}} z \approx_\eps z^\top L_G z~\forall z \in \K$. 

\begin{enumerate}
\item Initialize $s_0(e) = 1$ for $e\in E$. Let $t = 1$.
\item \textbf{While} $m_t := |\supp(s_{t-1})| > nf(n)/\eps$ for some fixed function $f(n)$ \textbf{do}
  \begin{enumerate}
  \item Find a partial coloring $x_t: E \to [-1,1]$ such that 
     \begin{enumerate}
     \item For any $z \in \K$,
        \begin{equation} \label{eq:vector}
          \Big|\sum\nolimits_e x_t(e) \cdot s_{t-1}(e) \cdot \inner{z}{b_e}^2\Big| \leq \frac{n f(n)}{m_t} \cdot z^\top L_G z
        \end{equation}      
     \item $| \{i \in \supp(s_{t-1}) \mid x_t(i) = \pm 1\}| = \Omega(m_t)$ and $x_t(i) = 0$ for all $i \not\in \supp(s_{t-1})$.
     \end{enumerate}
 \item If there are more $x_t(i) = 1$ than $x_t(i) = -1$ then update $x_t \gets - x_t$.
 \item Update $s_{t}(e)\gets s_{t}(e) = s_{t-1}(e) \cdot (1+x_t(e))$ for all $e \in E$.
  \item $t\gets t+1$.
\end{enumerate}
\item \textbf{Return:} $\Tilde{G}$ with edges weight $s_T(e)$ for $e\in E$ where $T$ is the last iteration.
\end{enumerate}
\end{framed}

\begin{lemma}\label{lemma:eps-sketch-algorithm}
Assuming that there is a deterministic polynomial time algorithm to find a partial coloring satisfying the requirements in Step 2(a) in each iteration,
then the deterministic spectral sketch algorithm is a deterministic polynomial time algorithm that outputs an $O(\eps)$-graphical spectral sketch with respect to $\K$ with $O(nf(n)/\eps)$ edges.
\end{lemma}
\begin{proof}
It is immediate that the graph $\Tilde{G}$ has only $O(n f(n) / \eps)$ edges by the design of the algorithm.
We thus focus on the spectral sketch guarantee.
Note that $z^\top L_{\Tilde{G}} z = z^\top \big( \sum_e s_T(e) \cdot b_eb_e^\top \big) z$ and $z^\top L_G z = z^\top \big( \sum_e s_0(e) \cdot b_eb_e^\top \big) z$ for each $z \in \K$.
Thus, by a telescoping sum and the triangle inequality,  
\[
\big|z^\top L_{\Tilde{G}}z - z^\top L_Gz\big| 
\leq \sum_t \Big|z^\top \Big( \sum_e x_t(e) \cdot s_{t-1}(e) \cdot b_eb_e^\top \Big) z\Big| 
\leq \sum_t \frac{nf(n)}{m_t} \cdot z^\top L_G z,
\]
where in the first inequality we used $s_t(e) - s_{t-1}(e) = x_t(e) \cdot s_{t-1}(e)$ by Step 2(c) of the algorithm, and in the second inequality we used the discrepancy bound \eqref{eq:vector} in Step 2(a)(i) of the algorithm.
Since $m_t$ is a geometrically decreasing sequence (which follows from Step 2(a)(ii), 2(b), and 2(c) as was explained in the proof of \autoref{theorem:sparsification-alg}), 
the sum on the right hand side is at most a constant factor of the last term, which is at most $\eps \cdot z^\top L_G z$ as $m_t > nf(n) / \eps$ in the final iteration. 
\end{proof}

Therefore, the spectral sketch problem is reduced to the partial coloring problem defined in~\eqref{eq:vector}, for which we need some new technical ideas as described below.

\subsubsection*{Vector Discrepancy Algorithm}

Note that the partial coloring in Step 2(a) is a vector discrepancy problem, rather than a matrix discrepancy problem as in \autoref{theorem:sparsification-alg}.
The problem of finding a low discrepancy signing $x\in \mathbb{R}^m$ to minimize its projection in a set of input directions is well studied in the literature.
To keep our algorithm deterministic, we will make use of the deterministic discrepancy minimization algorithm of Levy, Ramadas and Rothvoss in \cite{LRR17}.
We state their main result in a way that fits into the deterministic discrepancy walk algorithm\footnote{
We briefly sketch a simplistic version of the proof in~\cite{LRR17} to illustrate how that fits into the deterministic discrepancy walk algorithm. 
To bound the discrepancy, they used a potential function based on the multiplicative weights update method, which is defined as $\Phi(x) := \sum_{i=1}^m\Phi_i(x) := \sum_{i=1}^m \exp\big( \lambda \cdot \bigip{\frac{a_i}{\norm{a_i}}}{x} - \lambda^2\big)$. 
The subspace $U_t$ defined in \cite[Section 2]{LRR17} guarantees that as long as the update direction $y_t$ is chosen from $U_t$ at each iteration,
then (1) the total potential $\Phi(x_t)$ does not increase much~\cite[Lemma 9]{LRR17} and (2) each $\Phi_i(x_t)$ contributes at most a $O(1/m)$ fraction to the total potential~\cite[Lemma 10]{LRR17}. 
These conditions ensure that the final partial coloring satisfies the required discrepancy bound as stated. 
We remark that the constant $\frac45$ is arbitrary and can be changed to any larger constant at most one without changing the result.
We also remark that the framework in \autoref{s:partial-coloring} can also be used to recover a similar result in \autoref{theorem:LRR-main-thm} using the so called $\ell_q$-regularizer in~\cite{AZLO15,PV23}, but this derivation is omitted to keep the presentation simple.}.

\begin{theorem}[{\cite[Theorem 1 and Lemma 6]{LRR17}}] 
\label{theorem:LRR-main-thm}
Suppose we run the deterministic discrepancy walk algorithm with input vectors $a_1,\ldots,a_k \in \mathbb{R}^{m}$ where $k \geq m$. 
In each iteration $t$, there exists a polynomial time computable subspace $U_t$ of dimension at least $\frac45 m$ such that as long as the update direction $y_t\in U_t$ for all $t=1,2,\ldots,T$, 
then $x_T$ satisfies $|\inner{a_i}{x_T}|\leq O\big( \norm{a_i}_2 \sqrt{\log{\frac{k}{m}}} \big)$ for all $1 \leq i \leq k$.
\end{theorem}

The most straightforward way to apply \autoref{theorem:LRR-main-thm} to satisfy the vector discrepancy constraints in~\eqref{eq:vector} is to define an $|E|$-dimensional vector $a_z$ for each $z \in \K$ where $a_z(e) := s_{t-1}(e) \cdot \inner{z}{b_e}^2$ for each $e \in E$,
so that the constraints in~\eqref{eq:vector} become $|\inner{a_z}{x_t}| \leq \frac{nf(n)}{m_t} \cdot z^\top L_G z$ for each $z \in \K$ so as to apply \autoref{theorem:LRR-main-thm}.
The problem with this approach is that the guarantee in \autoref{theorem:LRR-main-thm} could be much larger than our desired bound\footnote{
To see this, consider the following example where $G$ is the complete graph,
and $\K := \{b_{i,j}\}_{i,j \in V}$ with one constraint for each pair of vertices.
Initially $s_0(e)=1$ for each $e \in E$.
Then, setting the input vectors $\{a_{i,j}\}_{i,j \in V}$ as above and applying \autoref{theorem:LRR-main-thm} gives a partial coloring $x$ with discrepancy bound 
roughly $\norm{a_{i,j}}_2 = \sqrt{\sum\nolimits_e \inner{b_{i,j}}{b_e}^4} \approx \sqrt{n}$ 
when $G$ is the complete graph.
On the other hand, the desired bound is $\frac{nf(n)}{m_t} \cdot z^\top L_G z \approx f(n)$ since $m \approx n^2$ and $b_{i,j}^\top L_G b_{i,j} = 2n-2$ in the complete graph.
We are interested in the regime where $f(n) = \polylog(n)$ and so the bound returned by \autoref{theorem:LRR-main-thm} using this straightforward approach is too weak.
}.
We remark that this straightforward approach can be used to obtain the bound $|\inner{a}{x_t}| \leq \sqrt{\frac{nf(n)}{m_t}} \cdot z^\top L_G z$, and this implies a spectral sketch of size $O(n f(n) / \eps^2)$ but that does not improve the spectral sparsification result~\cite{BSS12}.
The discrepancy bound in the form of~\eqref{eq:vector} is crucial for the sparsifier size to have a dependency of $1/\eps$ rather than $1/\eps^2$.

\subsubsection*{Main Ideas}

We explain the main ideas roughly before we present the formal proof below.
The starting observation is that by restricting $x$ to the degree-preserving subspace that satisfies $\sum_{u:u\sim v} x(u,v) \cdot s(u,v) = 0$ for all $v \in V$, we can rewrite the left hand side of the discrepancy constraints in~\eqref{eq:vector} as 
\[
\Big| \sum\nolimits_{uv\in E} x(u,v) \cdot s(u,v) \cdot \inner{z}{b_e}^2 \Big|
= 2\Big| \sum\nolimits_{uv\in E} x(u,v) \cdot s(u,v) \cdot z(u) \cdot z(v) \Big|.
\]
Then the idea is that we only need to sparsify edges from the dense parts of the graph, and so we freeze those edge variables that have large $s(u,v)$ or incident on low degree vertices.
By doing so, we get a useful bound that $s(u,v)^2 \lesssim \frac{n^2}{m^2} \cdot d(u) \cdot d(v)$ where $d(v)$ is the degree of vertex $v$ in the input graph $G$.
This will imply that the $2$-norm of the (rewritten) discrepancy constraint can be bounded by 
\[
\sqrt{\sum\nolimits_{uv \in E} s(u,v)^2 \cdot z(u)^2 \cdot z(v)^2} 
\lesssim \frac{n}{m} z^\top D z
\leq \frac{n}{\lambda m} z^\top L z,
\]
where the last inequality\footnote{This is not precise for the simplicity of the presentation.  We will see the correct version below.} follows when the input graph $G$ is a $\lambda$-expander.
This implies that we can apply \autoref{theorem:LRR-main-thm} to construct spectral sketches for $\Omega(1)$-expander graphs with only $O\big(\frac{n}{\eps} \sqrt{\log \frac{k}{m}} \big)$ edges.
Finally, we use the standard expander decomposition technique to construct spectral sketches for general graphs with $O(n \log^{3.5} n / \eps)$ edges.

\subsubsection*{Restricted Subspaces}

With the main ideas presented, we describe our restricted subspaces precisely.
In the $\tau$-th iteration, we would like to find a partial coloring $x_\tau \in \R^m$ given the current reweighting $s := s_{\tau-1} \in \R^m$, where $m$ is the number of edges in the input graph $G$.

\begin{enumerate}
\item
The first subspace is the degree-preserving subspace
\[
U^1 := \Big\{x\in \mathbb{R}^m ~\big|~ \sum\nolimits_{u: u\sim v} x(u,v) \cdot s(u,v) = 0 \quad \forall v\in V\Big\},
\]
which ensures that the weighted degrees of $s_\tau$ are the same as that in $s_{\tau-1}$ by Step 2(c) of the algorithm.
\item
The second subspace is to restrict to the dense parts of the graph.
For each vertex $v$, let $d_s(v) := | \{ e \in \supp(s) \mid v \in e \}$ be the degree of $v$ in the support of $s$, and $d(v)$ be the degree of $v$ in the input unweighted graph $G$.
Let $E_v := \{ uv \in E \mid s(u,v) > 10 \cdot d(v) / d_s(v) \}$ be the set of high weight edges incident on $v$.
Define $E_0 := \cup_{v \in V} E_v$.
Also let $V_{\rm low} := \{v \mid d_s(v) \leq \frac{1}{10} m/n\}$ be set of vertices with low degree in the support of $s$.
And define $E_1$ to be the set of edges incident to some vertices in $V_{\rm low}$.
The second subspace will be used to freeze the edges in $E_0 \cup E_1$:
\[
U^2 = \{x\in \mathbb{R}^m \mid x(e) = 0 \quad \forall e\in E_0\cup E_1\}.
\]
Since we maintain that $\sum_{u: u \sim v} s(u,v) = d(v)$ for all $v \in V$,
we see that $|E_v| \leq \frac{1}{10} \cdot d_s(v)$ by an averaging argument,
and thus $|E_0| \leq \frac{1}{10} \sum_{v \in V} d_s(v) = \frac15 m$.
And $|E_1| \leq \frac{1}{10} m$ as there are at most $n$ vertices.
This implies that $\dim((U^2)^\perp) \leq \frac{3}{10} m$, the complement of the subspace $U^2$ has dimension at most $\frac{3}{10} m$.
\item
The third set of subspaces is for the discrepancy constraints.
For each $z \in \K$, we first define a shifted vector $\bar{z}$ as follows.
Let $c_z := \sum_v d(v) \cdot z(v) / \sum_v d(v)$ be the center of $z$ with respect to the degrees.
Define $\bar{z} = z - c_z \vec{\one}$.
We note that $\bar{z}$ is constructed so that it satisfies the property $\sum_v d(v) \cdot \bar{z}(v) = 0$, which will be used in the spectral argument for expander graphs.
Now, for each $z \in \K$, define a vector $a_z \in \R^m$ where
\[
        a_z(u,v) = \begin{cases}
            &s(u,v) \cdot \bar{z}(u) \cdot \bar{z}(v) \quad {\rm~if~} uv \in E_s\\
            &0 \quad \text{otherwise}
        \end{cases}
\]
where $E_s := \supp(s) \setminus (E_0 \cup E_1)$.
Then the third set of subspaces is defined as
\[
U^3_t := {\rm~the~subspace~used~in~} t{\rm-th~iteration~in~\autoref{theorem:LRR-main-thm}~when~vectors~} \{a_z\}_{z \in \K} {\rm~are~given~as~input}.
\]
Note that there is one subspace for each iteration $t$ of the partial coloring algorithm.
\end{enumerate}

The final set of restricted subspace is defined as $\mathcal H_t := U^0 \cap U^1 \cap U^3_t$, for each iteration $t$ of the partial coloring algorithm.
By \autoref{theorem:LRR-main-thm}, $\dim(U^3_t) \geq \frac45 m$, and so
\[
\dim(\mathcal{H}_t) 
\geq \dim(U^3_t) - \dim((U^0)^\perp) - \dim((U^1)^\perp)
\geq \frac{4m}{5} - n - \frac{3m}{10}  = \frac{m}{2} - n,
\]
which is large enough as long as $m \geq \Omega(n)$.

\subsubsection*{Spectral Sketches for Expander Graphs}

We are ready to construct spectral sketches for expander graphs.

\begin{theorem}\label{corollary:expander-spectral-sketch}
Let $G$ be a $\lambda$-expander on $n$ vertices. 
For any set of vectors $\K$ with $|\K| \geq n$, 
there is a deterministic polytime algorithm to construct a deterministic $\eps$-spectral-sketch with $O\Big(\frac{n\sqrt{\log{(|\K|/n)}}}{\lambda\eps}\Big)$ edges.
\end{theorem}

\begin{proof}
We implement the partial coloring subroutine required by Step 2(a) of the deterministic graphical sketch algorithm with $f(n) = \frac{1}{\lambda} \sqrt{\log \frac{|\K|}{m}}$.
Suppose we are in the $\tau$-th iteration of the deterministic graphical sketch algorithm.
We let $s := s_{\tau-1}$ be the current reweighting, and let $x := x_{\tau}$ be the partial coloring that we would like to find.
For the partial coloring subroutine, we run a deterministic discrepancy walk algorithm (e.g.~the one in~\cite{LRR17}) with the restricted subspace $\mathcal{H}_t$ in the $t$-th iteration, and use the output of the algorithm as our partial coloring $x$ for this $\tau$-th iteration.
The main task is to check the discrepancy requirement in~\eqref{eq:vector} is satisfied.
For each $z \in \K$, since $x \in U^1$,
the left hand side of~\eqref{eq:vector} can be rewritten as
\begin{eqnarray*}
& & \sum_{e\in E} x(e) \cdot s(e) \cdot \inner{b_e}{z}^2 \\
& = & \sum_{uv \in E} x(u,v) \cdot s(u,v) \cdot (z(u)-  z(v) - c_z + c_z)^2\\
& = & \sum_{v\in V}(z(v)-c_z)^2\sum_{u:u\sim v} x(u,v) \cdot s(u,v) -  2\sum_{uv \in E} x(u,v) \cdot s(u,v) \cdot (z(u)-c_z) \cdot (z(v)-c_z)\\
& = & -2\sum_{uv\in E} x(u,v) \cdot s(u,v) \cdot \bar{z}(v) \cdot \bar{z}(v).
\end{eqnarray*}
Since the update direction $y_t$ in the partial coloring subroutine is in $U_t$ for all $t$, it follows that $x$ satisfies
\[
\bigg| \sum_{e\in E} x(e) \cdot s(e) \cdot \inner{z}{b_e}^2 \bigg| 
= 2\bigg| \sum_{uv\in E} x(u,v) \cdot s(u,v) \cdot \bar{z}(v) \cdot \bar{z}(v) \bigg|
= 2|\inner{a_z}{x}|
\lesssim \norm{a_z}_2 \sqrt{\log \frac{|\K|}{m}},
\]
where the second equality is by the definition of $a_z$ and $U^2$
and the inequality is by \autoref{theorem:LRR-main-thm}.

It remains to bound $\norm{a_z}_2$.
Note that for each edge $uv \in E_s$, 
it is not in $E_u$ or $E_v$ and so $s(u,v) \lesssim \min\big\{\frac{d(v)}{d_s(v)}, \frac{d(u)}{d_s(u)}\big\}$. 
Also, for each edge $uv \in E_s$, it is not in $E_1$ and so $d_s(u),d_s(v)\gtrsim \frac{m}{n}$. 
These imply that
\[
s(u,v)^2 
\lesssim \min\Big\{\frac{d(v)}{d_s(v)}, \frac{d(u)}{d_s(u)}\Big\}^2 
\leq \frac{d(u) \cdot d(v)}{d_s(u) \cdot d_s(v)} \lesssim \frac{n^2}{m^2} \cdot d(u) \cdot d(v).
\]
Then it follows from this and the definition of $a_z$ that
\[
\|a_w\|_2 
= \sqrt{\sum_{(u,v)\in E_s} s(u,v)^2 \cdot \bar{z}(u)^2 \cdot \bar{z}(v)^2} 
\lesssim \frac{n}{m}\sqrt{\sum_{(u,v)\in E_s} d(u) \cdot d(v) \cdot \bar{z}(u)^2 \cdot \bar{z}(v)^2}
\leq \frac{n}{m} \cdot \bar{z}^\top D \bar{z},
\]
where the last inequality holds as $(\bar{z}^\top D \bar{z})^2 = \sum_{u,v \in V} d(u) \cdot d(v) \cdot \bar{z}(u)^2 \cdot \bar{z}(v)^2$.
Finally, the Courant-Fischer theorem characterizes the second eigenvalue of the normalized Laplacian as 
\[
\lambda_2(\L_G) = \min_{x: \sum_v d(v) \cdot x(v) = 0} \frac{x^\top L_G x}{x^\top D x}.
\]
Since $\bar{z}$ was constructed to satisfy $\sum_v d(v) \cdot \bar{z}(v) = 0$,
it follows from Courant-Fischer that $\bar{z}^\top D \bar{z} \geq \lambda_2(\L_G) \cdot \bar{z}^\top L_G \bar{z} \geq  \lambda \cdot \bar{z}^\top L_G \bar{z} = \lambda \cdot z^\top L_G z$ as $\bar{z}$ is just a shift of $z$.
Putting together, we conclude that
\[
\bigg| \sum_{e\in E} x(e) \cdot s(e) \cdot \inner{z}{b_e}^2 \bigg|
\lesssim \norm{a_z}_2 \sqrt{\log \frac{|\K|}{m}}
\leq \frac{n}{m} \cdot \bar{z}^\top D \bar{z} \cdot \sqrt{\log \frac{|\K|}{m}} 
\leq \frac{n}{\lambda m} \cdot z^\top L_G z \cdot \sqrt{\log \frac{|\K|}{m}}. 
\]
Therefore, we have implemented a partial coloring algorithm required by Step 2(a) of the deterministic graphical sketch algorithm with $f(n) = \frac{1}{\lambda} \sqrt{\log \frac{|\K|}{m}}$, 
and thus the theorem follows from \autoref{lemma:eps-sketch-algorithm}.
\end{proof}

\subsubsection*{Proof of \autoref{theorem: deterministic-eps-sketch}}

To construct graphical spectral sketches for arbitrary unweighted undirected graphs, we once again use an expander-decomposition as was done in~\cite{CGP+23}.
By \autoref{fact: expander-decomp}, every graph $G$ can be decomposed into $G_1, \ldots, G_k$ such that each $G_i$ is an $\Omega(1 / \log^2{n})$-expander and each vertex is contained in at most $O( \log n)$ subgraphs. 
For each $G_i$, we apply \autoref{corollary:expander-spectral-sketch} to compute a deterministic $\eps$-graphical spectral sketch $\Tilde{G}_i$ of $G_i$, with 
$O\big(n_i \log^2{n}\sqrt{\log\frac{|\K|}{n}} / \eps\big)$ edges when $n_i := |V(G_i)|$.
Note that $\Tilde{G} := \sum_i \Tilde{G}_i$ is an $\eps$-graphical spectral sketch of $G$.
As each vertex in $G$ is contained in at most $\log{n}$ subgraphs, it follows that $\sum_{i=1}^k n_i = O(n \log n)$ and thus the total number of edges in $\Tilde{G}$ is $ O\big(n \log^3{n}\sqrt{\log\frac{|\K|}{n}} / \eps\big) $.

\subsection{Effective Resistance Sparsifiers} \label{s:resistance}

An interesting application of the graphical spectral sketch in \cite{CGP+23} was to construct $O( n \polylog n / \eps)$-sized effective resistance sparsifiers, answering an open question from \cite{ACK+16}.  
They proved the following useful property that relates sketches of Laplacian quadratic forms to sketches of its pseudo-inverse.

\begin{lemma}[{\cite[Lemma 6.8]{CGP+23}}] \label{lem:reff-sparsifier-conds}
Let $G = (V,E)$ be an undirected graph. 
Suppose $\Tilde{G}$ is a reweighted subgraph of $G$ satisfying
\begin{enumerate}
      \item $\Tilde{G}$ is am $O(\sqrt{\eps})$-spectral sparsifier of $G$,
      \item $(1 - \eps) \cdot b_{i,j}^\top L^{\dagger}_G b_{i,j} \leq b_{i,j}^\top L^\dagger_G L_{\Tilde{G}} L^{\dagger}_G b_{i,j} \leq (1+ \eps) \cdot b_{i,j}^\top L^{\dagger}_G b_{i,j}$ for all $i, j \in V$.
  \end{enumerate}
Then $(1 - \eps) \cdot b_{i,j}^\top L^{\dagger}_G b_{i,j} \leq b_{i,j}^\top L^{\dagger}_{\Tilde{G}} b_{i,j} \leq (1 + \eps) \cdot b_{i,j}^\top L^{\dagger}_G b_{i,j}$ for all $i, j \in V$.
 \end{lemma}
 
In other words, they proved that the problem of finding an effective resistance sparsifer can be reduced to finding a sparsifier that is simultaneously a $O(\sqrt{\eps})$-spectral sparsifier (which requires only $O(n/\eps)$ edges by \cite{BSS12}) and an $O(\eps)$-graphical spectral sketch with respect to the vectors $\{L^\dagger b_{i,j}\}_{i,j \in V}$.
They used the short cycle decomposition technique and expander decomposition to construct effective resistance sparsifiers with $O( n \polylog n / \eps )$ edges, but it is at least $\Omega( n \log^{16} n / \eps)$ even assuming optimal short cycle decomposition and optimal expander decomposition.

In this subsection, we show that the deterministic discrepancy framework can be used to construct sparser effective resistance sparsifiers and prove \autoref{cor:resistance-sparsifier}.
As was done in~\cite{CGP+23}, we first construct sparsifiers on expander graphs and then apply expander decomposition.

\begin{framed}{\textbf{Deterministic Effective Resistance Sparsification}} 

\textbf{Input:} a $\lambda$-expander $G = (V,E)$, and a target accuracy parameter $\eps$.

\textbf{Output:} $\Tilde{G} = (V, \Tilde{E})$ that is an $\eps$-effective resistance sparsifer of $G$ with $|\Tilde{E}| \lesssim n \sqrt{\log n} / (\lambda \eps)$. 

\begin{enumerate}
\item Initialize $s_0(e) = 1$ for $e\in E$. Let $t = 1$.
\item \textbf{While} $m_t := |\supp(s_{t-1})| > cn \sqrt{\log n} /(\lambda \eps)$ for some fixed constant $c$ \textbf{do}
\begin{enumerate}
\item Apply the matrix partial coloring algorithm with inputs $\{A_e\}_{e \in E}$ where $A_e = L_G^{\dagger/2}b_eb_e^\top L_G^{\dagger/2}$ and
vector constraints $\K = \{L_G^{\dagger}b_{i,j}\}_{i,j \in V}$.
Find a partial coloring $x_t: E \to [-1,1]$ such that
  \begin{enumerate}
  \item the matrix discrepancy is 
   \[\norm{\sum\nolimits_e x(e) \cdot s_{t-1}(e) \cdot L_G^{\dagger/2}b_eb_e^\top L_G^{\dagger/2}}_{\rm op} \lesssim \sqrt{\frac{n}{m}},\]
  \item the vector discrepancy is
   \[\sum\nolimits_{e}x(e) \cdot s_{t-1}(e) \cdot \big(b_{i,j}^\top L_G^{\dagger}b_e\big)^2 \lesssim \frac{n\sqrt{\log{n}}}{\lambda m}\cdot b_{i,j}^\top L^{\dagger}b_{i,j} \quad \forall i,j \in V,\]
  \item $x_t$ is degree-preserving such that $\sum_{u:u\sim v} x_t(u,v) \cdot s_{t-1}(u,v) = 0$ for all $v \in V$,
  \item $| \{i \in \supp(s_{t-1}) \mid x_t(i) = \pm 1\}| = \Omega(m_t)$ and $x_t(i) = 0$ for all $i \not\in \supp(s_{t-1})$.
  \end{enumerate}
\item If there are more $x_t(i) = 1$ than $x_t(i) = -1$ then update $x_t \gets - x_t$.
\item Update $s_{t}(e)\gets s_{t}(e) = s_{t-1}(e) \cdot (1+x_t(e))$ for all $e \in E$.
\item $t\gets t+1$
\end{enumerate}
\item \textbf{Return:} $\Tilde{G}$ with edges weight $s_T(e)$ for $e\in E$ where $T$ is the last iteration.
\end{enumerate}
\end{framed}

We prove the existence of a nearly-linear sized effective resistance sparsifers for $\Omega(1)$-expander graphs.

\begin{theorem} \label{t:Reff-expander}
Given a $\lambda$-expander graph $G$, the deterministic effective resistance sparsification algorithm always returns an $\eps$-effective resistance sparsifier with $O(\frac{n\sqrt{\log{n}}}{\lambda\eps})$ edges in deterministic polynomial time.
\end{theorem}
\begin{proof}
The main task is to show that the partial coloring in Step 2(a) can be implemented in deterministic polynomial time.
Suppose we are in the $\tau$-th iteration in the while loop.
Let $t \in [T]$ denote the iteration of a particular run of the deterministic walk partial coloring algorithm in Step 2(a).
In the proof of \autoref{lemma:deterministic-partial-coloring} for matrix discrepancy, we showed that there exists a subspace $U_{t}$ of dimension at least $\frac23 m_{t}$ such that if $y_{t}$ is picked from $U_{t}$ for all iterations $t$, then the matrix discrepancy requirement in Step 2(a)(i) is satisfied at the end\footnote{In the proof of \autoref{lemma:deterministic-partial-coloring}, we only showed that $\dim(U_t) \geq \frac14 m$, but this can be adjusted to have dimension at least $\frac23 m$ (or any arbitrary constant times $m$) by making the eigenspace in $U^3$ larger by a constant factor which would only increase the second-order term by a constant.}. 
In the proof of \autoref{corollary:expander-spectral-sketch} for spectral sketches, we showed that under the same framework, there exists a subspace $V_{t}$ of dimension at least $\frac12 m_{t} - n$ such that as long as $y_{t}$ is picked from $V_{t}$ for all iterations $t$, then the vector discrepancy requirement in Step 2(a)(ii) and the degree-preserving constraints in Step 2(a)(iii) are satisfied at the end. 
Therefore, by choosing $y_{t} \in U_{t}\cap V_{t}$, which still has large enough dimension as long as $m_t > 10n$ say, to incorporate the standard partial coloring constraints in Step 2(a)(iv).
    Finally, by using the same arguments as in  \autoref{theorem:sparsification-alg} and \autoref{lemma:eps-sketch-algorithm}, we see that at the end of the algorithm, our graph $\Tilde{G}$ has $O(\frac{n\sqrt{\log{n}}}{\lambda\eps})$ edges and satisfy that it is an $O(\sqrt{\eps})$-spectral sparsifer as well as an $O(\eps)$-spectral sketches with respect to $\K$, and hence $\Tilde{G}$ is an $\eps$ effective resistance sparsifier of $G$ by \autoref{lem:reff-sparsifier-conds}.
\end{proof}

Finally, we obtain \autoref{cor:resistance-sparsifier} using \autoref{t:Reff-expander} and the expander decomposition in \autoref{fact: expander-decomp}.
The proof is the same as in the proof of \autoref{theorem: deterministic-eps-sketch} and is omitted.

\section*{Concluding Remarks}

Building on recent works~\cite{RR20,PV23},
we developed a unified algorithmic framework for both discrepancy minimization and spectral sparsification, by combining the potential functions from spectral sparsification~\cite{BSS12,AZLO15} and the partial coloring and perturbation updates from discrepancy minimization~\cite{Spe85,Ban10,LM15,LRR17}.
We demonstrate this framework by showing simpler and improved constructions for various spectral sparsification problems.
The analysis is self-contained and elementary and is considerably simpler than that in~\cite{RR20}, and even in the standard setting is more intuitive and arguably simpler than that in~\cite{BSS12,AZLO15}.

Together with the results in~\cite{PV23}, this framework recovers best known results in many settings in discrepancy minimization and spectral sparsification, but not the most advanced ones such as the matrix Spencer problem~\cite{BJM23} and the Weaver's discrepancy problem~\cite{MSS15}.
It is thus an interesting and important open direction to extend this framework to recover these results.
One concrete and intermediate problem in this direction is to recover the result in~\cite{STZ24} for Eulerian sparsifiers, which uses the result in~\cite{BJM23} for matrix Spencer as a black box.
Other related problems are to recover some results in approximation algorithms~\cite{AO15,LZ22a,LZ22b}, which built on the results for spectral sparsification~\cite{AZLO15,AZLSW21} and for the Kadison-Singer problem~\cite{MSS15,AO14,KLS20}.

Another open question is to design fast deterministic algorithms for spectral sparsification and discrepancy minimization, which has been open even in the standard settings~\cite{BSS12,Ban10}.

\section*{Acknowledgements}

We thank Valentino Dante Tjowasi for collaborations at the initial stage of this work.

\bibliographystyle{alpha}
\bibliography{references}

\appendix

\section{Simple Proof of \autoref{l:potential}} \label{a:omitted}

We first compute the optimizer $M$.
Since $\Delta_m$ is a compact set and the regularizer $\phi(M)=-2\tr(M^\frac12)$ is strongly convex, the optimizer of \eqref{eq:potential} is attained and uniquely defined. 
Moreover, as the gradient of the $\ell_{1/2}$-regularizer is $\nabla \phi(M) = M^{-1/2}$, which blows up when $M$ is singular, the optimizer of \eqref{eq:potential} stays in the interior of $\Delta_m$. 
Therefore, we can apply the KKT condition to \eqref{eq:potential} without the constraint $M \succcurlyeq 0$ (so only the constraint $\tr(M) = \inner{M}{I} = 1$), and obtain a closed-form characterization of the unique optimizer
\[ M = (u_x I_n - \eta A(x))^{-2}, \]
where $u_x$ is the unique value such that $M$ is a density matrix (note that $u_x$ is a function of $x$).
We can then compute the closed-form characterization of the potential function. Let $A(x) = \sum_{i=1}^n \lambda_i v_i v_i^\top$ be its eigen-decomposition. 
Then $M = \sum_{i=1}^n (u_x - \eta \lambda_i)^{-2} v_i v_i^\top$, and so
\begin{align*}
\Phi(x) 
= \frac{1}{\eta} \tr(M^{\frac12}) + \langle A(x), M \rangle + \frac{1}{\eta}\tr(M^{\frac12}) 
= \frac{1}{\eta} \tr(M^{\frac12}) + \sum_{i=1}^n \left( \lambda_i (u_x - \eta \lambda_i)^{-2} + \frac{(u_x - \eta \lambda_i)^{-1}}{\eta} \right).
\end{align*}
Rearranging the terms and using $\sum_{i=1}^n (u_x - \eta \lambda_i)^{-2} = 1$ as $M$ is a density matrix, we obtain that
\begin{align} \label{eq:Phi}
    \Phi(x) = \frac{1}{\eta} \tr(M^{1/2}) + \frac{u_x}{\eta} \cdot \sum_{i=1}^d (u_x - \eta \lambda_i)^{-2} = \frac{1}{\eta} \tr((u_x I_n - \eta A(x))^{-1}) + \frac{u_x}{\eta}. 
\end{align}
Note that this is basically the same as the trace inverse potential function in~\cite{BSS12}.

To bound the potential increase, the following claim was implicitly proved in~\cite{AZLO15} using mirror descent and Bregman divergence. 
Here we provide a simpler proof using elementary convexity arguments.

\begin{claim} \label{cl:trace-inverse}
If $\norm{(u_x I_n - \eta A(x))^{-1} \cdot \eta A(y)}_{\rm op} < 1$, then it holds that
\[
\Phi(x+y) - \Phi(x) \leq \frac{1}{\eta} \left( \tr\big((u_x I_n - \eta A(x+y))^{-1}\big) - \tr\big((u_x I_n - \eta A(x))^{-1}\big) \right).
\]
\end{claim}

\begin{proof}
By \eqref{eq:Phi}, the potential increase is
\begin{align} \label{eq:change}
    \Phi(x + y) - \Phi(x) = \frac{1}{\eta} \Big( \underbrace{\tr\big( (u_{x+y} I_n - \eta A(x+y))^{-1} \big)}_{(*)} - \tr\big( (u_x I_n - \eta A(x))^{-1} \big) \Big) + \frac{u_{x+y} - u_x}{\eta}.
\end{align}

Given $\eta$ and $x$ and $y$, consider the univariate function
\[
    f(u) = \tr\big( (u I_n - \eta A(x+y))^{-1} \big).
\]
Note that the $(*)$ term is equal to $f(u_{x+y})$.
We will use the convexity of $f$ to bound the $(*)$ term.

First, we show that $f(u)$ is convex within the interval between $u_x$ and $u_{x+y}$.
Let $\beta_1 \leq \cdots \leq \beta_n$ be the eigenvalues of $\eta A(x+y)$. 
When $u > \beta_n = \lambda_{\max}(\eta A(x+y))$, we can check that $f(u) = \sum_{i=1}^n (u-\beta_i)^{-1}$ is convex in $u$ (say by the second derivative test).
So, we just need to show that both $u_{x+y}$ and $u_x$ are greater than $\lambda_{\max}(\eta A(x+y))$ to establish the convexity within the interval between them.
For the former, since $(u_{x+y} I_n - \eta A(x+y))^{-2}$ is positive definite, it follows immediately that $u_{x+y} > \lambda_{\max}(\eta A(x+y))$.
For the latter, the assumption $\norm{(u_x I_n - \eta A(x))^{-1} \cdot \eta A(y)}_{\rm op} < 1$ implies that $\eta A(y) \preccurlyeq u_x I_n - \eta A(x)$, and thus $u_x > \lambda_{\max}(\eta A(x+y))$.
Therefore, $f(u)$ is convex within the interval between $u_x$ and $u_{x+y}$.
This implies that
\[
f(u_{x+y}) + f'(u_{x+y}) (u_x - u_{x+y}) \leq f(u_x),
\]
or equivalently,
\begin{align} \label{eq:convexity}
\tr\big( (u_{x+y} I_n - \eta A(x+y))^{-1} \big) \leq \tr\big( (u_x I_n - \eta A(x+y))^{-1} \big) - f'(u_{x+y}) \cdot (u_x - u_{x+y}).
\end{align}
It remains to compute $f'(u_{x+y})$.
The derivative of $f(u)$ when $u > \lambda_{\max}(\eta A(x+y))$ is 
\begin{align*}
    f'(u) & = \tr(\partial_u (u I_n - \eta A(x+y))^{-1}) 
     = - \tr\left( (u I_n - \eta A(x+y))^{-2} \right).
\end{align*}
Since $(u_{x+y} I_n - \eta A(x+y))^{-2}$ is a density matrix, when $f'(u)$ is evaluated at $u_{x+y}$, we have
\begin{align*}
    f'(u_{x+y}) = - \tr\left( (u_{x+y} I_n - \eta A(x+y))^{-2} \right) = -1.
\end{align*}
Plugging it into \eqref{eq:convexity}, it follows that
\begin{align*}
\tr\big( (u_{x+y} I_n - \eta A(x+y))^{-1} \big) \leq \tr\big( (u_x I_n - \eta A(x+y))^{-1} \big) + (u_x - u_{x+y}).
\end{align*}
Combining with \eqref{eq:change}, we conclude that
\[
    \Phi(x + y) - \Phi(x) \leq \frac{1}{\eta} \left( \tr\big( (u_x I_n - \eta A(x+y))^{-1} \big) - \tr\big( (u_x I_n - \eta A(x))^{-1} \big) \right). \qedhere
\]
\end{proof}

We now apply \autoref{l:RR2} (from~\cite{RR20}) to finish the proof.
Set $A := u_x I_n - \eta A(x) = M^{-\frac12} \succ 0$ and $B:=A(y)$.
The assumption $\norm{M^{\frac12} \cdot \eta A(y)}_{\rm op} \leq \frac12$ in \autoref{l:potential} translates to the assumption $\norm{\eta A^{-1} B}_{\rm op} \leq \frac12$ in \autoref{l:RR2}.
So, we can apply \autoref{l:RR2} to conclude that there is a value $c \in [-2,2]$ so that
\[
\tr((u_x I_n - \eta A(x+y))^{-1}) - \tr((u_x I_n - \eta A(x))^{-1}) = \eta \tr(M A(y)) + c \eta^2 \tr(M^{\frac12} A(y) M^{\frac12} A(y) M^{\frac12}).
\]
Combining with \autoref{cl:trace-inverse}, we proved \autoref{l:potential}.

\end{document}